\documentclass[journal]{IEEEtran}

\ifCLASSINFOpdf
\else
\fi

\pdfoutput=1

\usepackage{amssymb}
\usepackage{amsthm}
\usepackage[cmex10]{amsmath}
\usepackage{amsfonts}
\usepackage{dsfont}
\usepackage{graphicx}
\usepackage[caption=false]{subfig}
\usepackage{cite}
\usepackage{url}
\usepackage{algorithm}
\usepackage{algorithmic}
\usepackage{color}

\def\be{\begin{equation}}
\def\ee{\end{equation}}
\def\ben{\begin{equation*}}
\def\een{\end{equation*}}
\def\bea{\begin{eqnarray}}
\def\eea{\end{eqnarray}}
\def\beaa{\begin{eqnarray*}}
\def\eeaa{\end{eqnarray*}}
\def\biea{\begin{IEEEeqnarray}{rCl}}
\def\eiea{\end{IEEEeqnarray}}

\def\mb{\mathbf}
\def\bs{\boldsymbol}
\def\ds{\displaystyle}

\def\RCMLEL{$\text{RCML}_\text{EL}\;$}

\def\RCMLC{$\text{RCML}_\text{Chen}\;$}
\def\RCMLAIC{$\text{RCML}_\text{AIC}\;$}
\def\CNCEL{$\text{CNCML}_\text{EL}\;$}
\def\CNCML{$\text{CNCML}\;$}

\DeclareMathOperator{\rank}{rank}
\DeclareMathOperator{\sinc}{sinc}
\DeclareMathOperator{\lr}{LR}
\DeclareMathOperator{\tr}{tr}
\DeclareMathOperator{\diag}{diag}
\DeclareMathOperator{\Kmax}{K_{\max}}

\theoremstyle{remark} \newtheorem{lemma}{Lemma}
\theoremstyle{remark} 

\begin{document}

\title{Robust Covariance Estimation under Imperfect Constraints using an Expected Likelihood Approach}

\author{Bosung~Kang,~\IEEEmembership{Student~Member,~IEEE,}
        Vishal~Monga,~\IEEEmembership{Senior~Member,~IEEE,}
        Muralidhar~Rangaswamy,~\IEEEmembership{Fellow,~IEEE,}
        and Yuri~Abramovich,~\IEEEmembership{Fellow,~IEEE}
        \thanks{Research was supported by AFOSR grant number FA9550-12-1-0333.

        Dr. Rangaswamy was supported by the Air Force Office of Scientific
Research under project 2311IN.}% <-this % stops a space
        }

\maketitle

\vspace{-10mm}

\begin{abstract}
We address the problem of structured covariance matrix estimation for radar space-time adaptive processing (STAP). A priori knowledge of the interference environment has been exploited in many previous works to enable accurate estimators even when training is not generous. Specifically, recent work has shown that employing practical constraints such as the rank of clutter subspace and the condition number of disturbance covariance leads to powerful estimators that have closed form solutions. While rank and the condition number are very effective constraints, often practical non-idealities makes it difficult for them to be known precisely using physical models. Therefore, we propose a robust covariance estimation method for radar STAP via an expected likelihood (EL) approach. We analyze covariance estimation algorithms under three cases of imperfect constraints: 1) a rank constraint, 2) both rank and noise power constraints, and 3) condition number constraint. In each case, we formulate precise constraint determination as an optimization problem using the EL criterion. For each of the three cases, we derive new analytical results which allow for computationally efficient, practical ways of setting these constraints. In particular, we prove formally that both the rank and condition number as determined by the EL criterion are unique. Through experimental results from a simulation model and the KASSPER data set, we show the estimator with optimal constraints obtained by the EL approach outperforms state of the art alternatives.
\end{abstract}


\IEEEpeerreviewmaketitle

\section{Introduction}
\label{Sec:Introduction}
Radar systems using multiple antenna elements and processing multiple pulses are widely used in modern radar signal processing since it helps overcome the directivity and resolution limits of a single sensor. Joint adaptive processing in the spatial and temporal domains for the radar systems, called space-time adaptive processing (STAP) \cite{Guerci03,Klemm02,Monzingo04}, enables suppression of interfering signals as well as preservation of gain on the desired signal. Interference statistics, in particular the covariance matrix of the disturbance, which must be estimated from secondary training samples in practice, play a critical role on the success of STAP. To obtain a reliable estimate of the disturbance covariance matrix, a large number of homogeneous training samples are necessary. This gives rise to a compelling challenge for radar STAP because such generous homogeneous (target free) training is generally not available in practice \cite{Himed97}.

Much recent research for radar STAP has been developed to overcome this practical limitation of generous homogeneous training. Specifically, the knowledge-based processing which uses \emph{a priori} information about the interference environment is widely referred in the literature \cite{Guerci06,Wicks06} and has merit in the regime of limited training data. These techniques include intelligent training selection \cite{Guerci06} and the spatio-temporal degrees of freedom reduction \cite{Wicks06,Wang91,Gini08}. In addition, covariance matrix estimation techniques that enforce and exploit a particular structure have been pursued as one approach of these methods. Examples of structure include persymmetry \cite{Nitzberg80}, Toeplitz structure \cite{Li99,Fuhrmann91,Abramovich98}, circulant structure \cite{Conte98}, eigenstructure \cite{Steiner00,Kang14,Aubry12}. In particular, the fast maximum likelihood (FML) method \cite{Steiner00} which enforces a special eigenstructure that the disturbance covariance matrix represents a scaled identity matrix plus a rank deficient and positive semidefinite clutter component also falls in this category and is shown to be the most competitive technique experimentally.

Previous works, notably in statistics \cite{Anderson63,Wax85} (and references therein) have considered factor analysis approaches for incorporating rank information in ML estimation. Recently, Kang \emph{et al.} \cite{Kang14} have developed extensions based on convex optimization approaches and furnished closed forms for rank constrained ML (RCML) estimation in practial radar STAP. Crucially, Kang \emph{et al.} show that rank of the clutter covariance if exactly known and incorporated, enables much higher normalized SINR and detection performance over the state-of-the-art, particularly FML, even under limited training.

Aubry \emph{et al.} \cite{Aubry12} also improve upon the FML by exploiting a practical constraint inspired by physical radar environment, specifically the eigenstructure of the disturbance covariance matrix. They employed a condition number of the interference covariance matrix as well as the structural constraint used in the FML. Though the initial optimization problem is non-convex, the estimation problem is reduced to a convex optimization problem.

In \cite{Kang14}, the authors assume the rank of the clutter is given by Brennan rule \cite{Ward94} under ideal conditions. However, in practice (under non-ideal conditions) the clutter rank departs from the Brennan rule prediction due to antenna errors and internal clutter motion. In this case, the rank is not known precisely and needs to be determined before using the RCML estimator. Determination of the number of signals in a measurement record is a classical eigenvalue problem, which has received considerable attention in the past 60 years. It is important to note that the problem does not have a simple and unique solution. Consequently, a number of techniques have been developed to address this problem \cite{Akaike74,Rissanen78,Wax85,Yin87,Tufts94}. The problem of rank estimation using the knowledge aided sensor signal processing and expert reasoning (KASSPER) data \cite{Bergin02} was also studied in \cite{Abramovich11} for the time varying multichannel autoregressive model, that provides an approximation to the spectral properties underlying the clutter phenomenon. A detailed comparison of the approach adopted here with that of \cite{Abramovich11} is beyond the scope of this paper. The condition number is also rarely known precisely, in fact Aubry {\em et al.} \cite{Aubry12} employ an ML estimate of the condition number.

Expected likelihood (EL) approach \cite{Abramovich07} has been proposed to determine a regularization parameter based on the statistical invariance property of the likelihood ratio (LR) values. Specifically, the probability density function (pdf) of LR values for the true covariance matrix depends on only the number of training samples ($K$) and the dimension of the true covariance matrix ($N$), not the true covariance itself under a Gaussian assumption on the observations. This statistical independence of LR values on the true covariance itself enables pre-calculation of LR values even though the true covariance is unknown. Finally, the regularization parameters are selected so that the LR value of the estimate agrees as closely as possible with the {\em median} LR value determined via its pre-characterized pdf.

\textbf{Contributions:} In view of the aforementioned observations, we develop covariance estimation methods which automatically and adaptively determine the values of practical constraints via an expected likelihood approach for practical radar STAP.\footnote{A preliminary version of the work appeared at the 2015 IEEE Radar Conference \cite{Kang15Radarcon}.} Our main contributions are outlined below.
\begin{itemize}
    \item \textbf{Fast Algorithms for adaptively determining practical constraints:} We propose methods to select practical constraints employed in the optimization problems for covariance estimation in radar STAP using the expected likelihood approach. The proposed methods guide the selection of the constraints via the expected likelihood criteria when they are imperfectly known. We consider three different cases of the constraints in this paper: 1) the clutter rank constraint, 2) jointly the rank and the noise power constraints, and 3) the condition number constraint.
    \item \textbf{Analytical results with formal proofs:} For each case mentioned above, we derive {\em new} analytical results. We first formally prove that the rank selection problem based on the expected likelihood approach has a unique solution. This guarantees there is only one rank which is the best (global optimal) rank in the sense of the EL approach. Second, we derive a closed-form solution of the optimal noise power for a given rank, which means we do not need iterative or numerical methods to find the optimal noise power, which in turn enables fast implementation. Finally, we also prove there exists a unique optimal condition number for the condition number selection criterion via the EL approach.
    \item \textbf{Experimental Results through simulated model and the KASSPER data set:} Experimental investigation on a simulation model and on the KASSPER data set shows that the proposed methods for three different cases outperform alternatives such as the FML, leading rank selection methods in radar literature and statistics, and the ML estimation of the condition number constraint with respect to the normalized output SINR.
\end{itemize}

The rest of the paper is organized as follows. Section \ref{Sec:Background} briefly reviews the previous structured covariance estimation methods including the rank constrained ML estimator and the condition number constrained ML estimator and the expected likelihood approach. Constraint selection problems via the EL approach and their corresponding solutions are provided in Section \ref{Sec:Proposed}. Section \ref{Sec:Experiments} performs experimental validation wherein we report the performance of the proposed method and compare it against existing methods in terms of normalized output SINR on both the simulation model and the KASSPER data set. Section \ref{Sec:Conclusion} concludes the paper.

\section{Background}
\label{Sec:Background}

In this section, we briefly provide a review of related structured covariance estimation algorithms and the expected likelihood criterion which can be useful in estimating parameters/constraints.

%Recently, structured covariance matrix estimation methods which employ inherent parameters in the covariance matrix, in particular the rank of clutter and the condition number of covariance, as constraints of the optimization problem have been proposed for radar STAP and it has been shown that they successfully works in radar applications. For both algorithms, the rank constraint ML estimator and the condition number constraint ML estimator, the closed form solutions have been derived although the initial optimization problems are non-convex.
%
%The expected likelihood criterion on which our proposed methods are based has been proposed by Abramovich \emph{et al.} It enables to find and refine parameters so that the value of likelihood ratio of the estimate can be as close to that of the true covariance as possible whereas the LR value of the maximum likelihood estimate is still far from that of the true covariance.

%\subsection{Problem formulation}

\subsection{Rank Constrained ML estimation}

It has been shown \cite{Kang14} that the rank can be employed into the optimization problem in a tractable manner and the RCML estimator is the best STAP estimator when the rank is accurately predicted by the Brennan rule. The initial non-convex optimization problem for the rank constrained ML estimation is given by
\be
\label{Eq:InitialProblemRCML}
\left\{ \begin{array}{cc}
\ds \max_{\mb R} & f(\mb Z) = \frac{1}{\pi^{NK}|\mb R|^K}\exp(-\tr\{\mb Z^H \mb R^{-1} \mb Z\})\\
s.t. & \mb{R} = \sigma^2 \mb{I} + \mb{R}_c\\
 & \rank(\mb{R}_c) = r\\
 & \mb{R}_c \succeq \mb 0 \end{array} \right.
\ee

Rank constrained ML estimation has been studied in statistics \cite{Anderson63} and in the radar signal processing literature \cite{Kang14}. In particular, the closed form estimator when the radar noise floor is known is given by \cite{Kang14}
\be
\label{Eq:RCMLsolution2}
\mb{R}^\star = \sigma^2 {\mb{X}^\star}^{-1} = \sigma^2 \mb{V}{\mb\Lambda^\star}^{-1}\mb{V}^H
\ee
where $\mb{V}$ is the eigenvector matrix of the sample covariance matrix $\mb{S}$ and $\mb\Lambda^\star$ is a diagonal matrix with diagonal entries $\lambda_i^\star$ which is given by
\be
\label{Eq:RCMLsolution}
\lambda_i^\star = \left\{ \begin{array}{cc}
\min (1,\dfrac{1}{\bar d_i}) & \text{for} \; i=1,2,\ldots,r\\
1 & \text{for} \; i=r+1,r+2,\ldots,N \end{array} \right.
\ee
where $\bar d_i$'s are the eigenvalues of the normalized sample covariance and $r$ is the clutter rank. Note that the ML solution of the eigenvalue is a function of the rank $r$ and $\bar d_i$'s.

\subsection{Condition Number constrained ML estimation}

Aubry \emph{et al.} proposed the method of a structured covariance matrix under a condition number upper-bound constraint \cite{Aubry12}. The initial non-convex optimization problem is
\be
\label{Eq:InitialProblemCNCML}
\left\{ \begin{array}{cc}
\ds \max_{\mb R} & f(\mb Z) = \frac{1}{\pi^{NK}|\mb R|^K}\exp(-\tr\{\mb Z^H \mb R^{-1} \mb Z\})\\
s.t. & \mb{R} = \sigma^2 \mb{I} + \mb{R}_c\\
 & \frac{\lambda_{\max}(\mb R)}{\lambda_{\min} (\mb R)} \leq K_{\max}\\
 & \mb{R}_c \succeq \mb 0\\
 & \sigma^2 \geq c \end{array} \right.
\ee

The authors showed that the optimization problem falls within the class of MAXDET problems \cite{Vandenberghe98,DeMaio09} and developed an efficient procedure for its solution in closed form which is given by

\be
\label{Eq:CN1}
\mb{R}^\star = \mb{V}{\mb\Lambda^\star}^{-1}\mb{V}^H
\ee
where
\be
\mb\Lambda^\star = \diag\big(\bs\lambda^\star (\bar u)\big)
\ee
, $\bs\lambda^\star (\bar u) = [\lambda_1^\star (\bar u),\ldots,\lambda_N^\star (\bar u)]$ with
\be
\lambda_i^\star (\bar u) = \min \bigg(\min (K_{\max} \bar u, 1), \max \Big(\bar u,\frac{1}{\bar d_i}\Big) \bigg)
\ee
, $K_{\max}$ is a condition number constraint, and $\bar u$ is an optimal solution of the following optimization problem,
\be
\label{Eq:OptimizationU}
\left\{ \begin{array}{cc}
\ds\min_u & \sum_{i=1}^N G_i(u)\\
s.t. & 0 < u \leq 1 \end{array} \right.
\ee
where
\be
\label{Eq:Giu1}
G_i(u) = \left\{ \begin{array}{ll}
\log K_{\max} - \log u + K_{\max} \bar d_i u & \text{if} \quad 0 < u \leq \frac{1}{K_{\max}}\\
\bar d_i & \text{if} \quad \frac{1}{K_{\max}} \leq u \leq 1 \end{array} \right.
\ee
for $\bar d_i \leq 1$, and
\biea
\label{Eq:Giu2}
\lefteqn{G_i(u)}\nonumber\\ & = & \left\{ \begin{array}{ll}
\log K_{\max} - \log u + K_{\max} \bar d_i u & \text{if} \quad 0 < u \leq \frac{1}{K_{\max} \bar d_i}\\
\log \bar d_i + 1 & \text{if} \quad \frac{1}{K_{\max} \bar d_i} < u \leq \frac{1}{\bar d_i}\IEEEeqnarraynumspace\\
\frac{1}{\bar d_i} + \bar d_i u & \text{if} \quad \frac{1}{\bar d_i} \leq u \leq 1 \end{array} \right.
\eiea
for $\bar d_i > 1$. Similar to the RCML estimator, the ML solution is a fucntion of $\bar d_i$'s and the condition number $K_{\max}$.

\subsection{Expected Likelihood Approach}
Abramovich \emph{et al.} \cite{Abramovich07} proposed an approach called the expected likelihood (EL) method which develops a new criterion for selection of parameters such as the loading factor based on direct likelihood matching. Expected likelihood approach is motivated by invariance properties of the likelihood ratio (LR) value which is given by
\bea
\lr(\mb R, \mb Z) & \equiv & \Big[\dfrac{f(\mb Z | \mb R)}{f(\mb Z | \mb S)}\Big]^{1/K}\\
& = & \frac{|\mb R^{-1} \mb S| \exp N}{\exp[\tr(\mb R^{-1} \mb S)]}
\eea
under a Gaussian assumption on the observations, $\mb z_i$'s. Furthermore, the unconstrained ML solution $\mb S$ has the LR value of 1. That is,
\be
\ds\max_{\mb R} \lr (\mb R, \mb Z) = \lr (\mb S, \mb Z) = 1
\ee

However, as shown in \cite{Abramovich07} the LR values of the true covariance matrix $\mb R_0$ are much lower than that of the ML solution $\mb S$. Therefore, it seems natural to replace the ML estimate by one that generates LR values consistent with what is expected for the true covariance matrix. More importantly, Abramovich \emph{et al.} showed \cite{Abramovich07} that the pdf of the LR for the true covariance matrix, which is given by
\bea
\lr (\mb R_0, \mb Z) & = & \frac{|\mb R_0^{-1} \mb S| \exp N}{\exp[\tr(\mb R_0^{-1} \mb S)]}\\
& = & \frac{|\mb R_0^{-1/2} \mb S \mb R_0^{-1/2}| \exp N}{\exp[\tr(\mb R_0^{-1/2} \mb S \mb R_0^{-1/2})]}
\eea
does not depend on the true covariance itself since
\be
\hat{\mb C} \equiv N \mb R_0^{-1/2} \mb S \mb R_0^{-1/2} \sim \mathcal{CW}(K,N,\mb I)
\ee
where $\mathcal{CW}$ represents complex Wishart distribution which is determined entirely by $K$ and $N$ and does not need $\mb R_0$. Therefore, the pdf of LR values for the true covariance matrix can be precalculated for given $K$ and $N$ and indeed the moments of distribution of the LR values were derived by Abramovich \emph{et al.} in their paper \cite{Abramovich07}.

Based on the invariance of the pdf of LR values, the EL approach can be used to determine values of parameters in estimation problems. For instance, the EL estimator for a diagonally loaded SMI technique under homogeneous interference training conditions and fluctuating target with known power is given by \cite{Abramovich07}
\be
\hat{\mb R}_{\text{LSMI}} = \hat\beta \mb I + \mb S
\ee
where
\be
\label{Eq:OptimalBeta}
\hat\beta \equiv \ds\arg_\beta \Bigg\{ \dfrac{|(\beta \mb I + \mb S)^{-1} \mb S| \exp N}{\exp \big( \tr[(\beta \mb I + \mb S)^{-1} \mb S] \big)} \equiv \lr_0 \Bigg\}
\ee
and $\lr_0$ is the reference median statistic, which can be precalculated from the pdf of the LR values
\be
\int_0^{\lr_0} f\big[\lr(\mb R_0, \mb Z)\big] d\lr = 0.5
\ee
where $f\big[\lr(\mb R_0, \mb Z)\big]$ is the invariant pdf of the LR values.

\section{Constraints selection method via Expected Likelihood Approach}
\label{Sec:Proposed}

\subsection{Selection of rank constraint}
\label{Sec:Rankonly}

\begin{figure}
\centering
\includegraphics[scale=0.43]{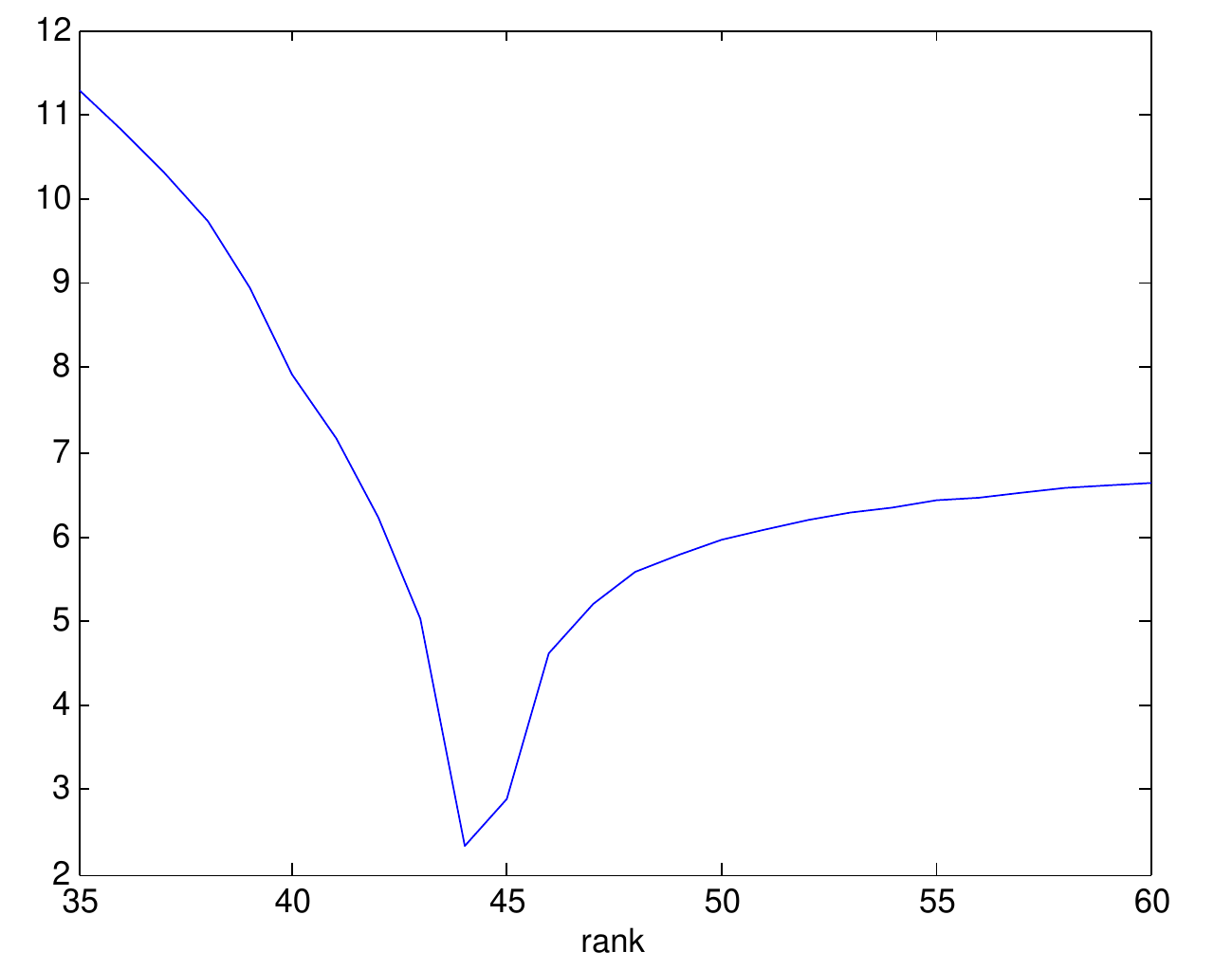}
\caption{$\bigg( \log\Big( \lr\big(\mb R_\text{RCML}(r), \mb Z\big)/\lr_0\Big)\bigg)^2$ versus $r$ for KASSPER dataset ($K=2N=704$)}
\label{Fig:LRdifference}
\end{figure}

%In Section \ref{Sec:Background}, we discuss that the RCML estimator is not only powerful in practice but also computationally cheap and the EL approach is useful to select parameters so that the estimate is consistent with the true covariance matrix in the sense of the LR value. From Eq. \eqref{Eq:RCMLsolution}, we see the RCML solution is a function of the rank $r$ and $d_i$'s which are given in the problem.

We propose to use the EL approach to refine and find the {\em optimal} rank when the rank determined by underlying physics is not necessarily accurate.

Now we set up the optimization criterion to find the rank via the EL approach. Since the rank is an integer, there may not exist the rank which exactly satisfies Eq. \eqref{Eq:OptimalBeta}. Therefore, we instead find a rank which such that the corresponding LR value departs the least from the median (and precomputed) LR value $\lr_0$. That is,
\be
\hat{\mb R}_{\text{RCML}_\text{EL}} = \sigma^2 \mb{V} {\mb\Lambda^\star}^{-1}(\hat r) \mb{V}^H
\ee
where
\be
\label{Eq:OptimalRank}
\hat r \equiv \arg\min_{r \in \mathds{Z}} \Big| \lr\big(\mb R_\text{RCML}(r), \mb Z\big)  - \lr_0 \Big|^2
\ee
and $\lr\big(\mb R_\text{RCML}(r), \mb Z\big)$ is given by Eq. \eqref{Eq:LR_RCML1}.

Now we investigate the optimization problem \eqref{Eq:OptimalRank} for the rank selection. Since the eigenvectors of $\mb R_\text{RCML}$ are identical to those of the sample covariance matrix $\mb S$ as shown in Eq. \eqref{Eq:RCMLsolution2}, the LR value of $\mb R_\text{RCML}$ in Eq. \eqref{Eq:OptimalRank} can be reduced to the function of the eigenvalues of $\mb R_\text{RCML}$ and $\mb S$. Let the eigenvalues of $\mb R_\text{RCML}$ and $\mb S$ be $\lambda_i$ and $d_i$ (arranged in descending order). Then the LR value of $\mb R_\text{RCML}$ can be simplified to a function of ratio of $d_i$ to $\lambda_i$, $\dfrac{d_i}{\lambda_i}$. That is,
\bea
\lr\big(\mb R_\text{RCML}(r), \mb Z\big) & = & \dfrac{|\hat{\mb R}_{\text{RCML}}^{-1}(r) \mb S| \exp N}{\exp \Big( \tr\big[\hat{\mb R}_{\text{RCML}}^{-1}( r) \mb S\big] \Big)}\label{Eq:LR_RCML1}\\
& = & \frac{\ds\prod_{i=1}^N \dfrac{d_i}{\lambda_i} \cdot \exp N}{\exp\Big[\ds\sum_{i=1}^N \dfrac{d_i}{\lambda_i}\Big]}\label{Eq:SimplifiedLR}
\eea
\begin{lemma}
\label{Lemma1}
The LR value of the RCML estimator, $\lr\big(\mb R_\text{RCML}(r), \mb Z\big)$, is a monotonically increasing function with respect to the rank $r$ and there is only one unique $\hat r$ in the optimization problem \eqref{Eq:OptimalRank}.
\end{lemma}
\begin{proof}
We derive the relationship between $\lr\big(\mb R_\text{RCML}(i)\big)$ and $\lr\big(\mb R_\text{RCML}(i+1)\big)$. See Appendix A for details.
\end{proof}

Lemma \ref{Lemma1} gives us a significant analytical result that is the EL approach leads to a unique value of the rank, i.e., when searching over the various values of the rank it is impossible to come up with multiple choices. That also means that it is guaranteed that we can always find the global optimum of $r$ not local optima (minima).  We plot the values of $\bigg( \log\Big( \lr\big(\mb R_\text{RCML}(r), \mb Z\big)/\lr_0\Big)\bigg)^2$ versus the rank $r$ for one realization for the KASSPER dataset ($K=2N=704$) in Fig. \ref{Fig:LRdifference}. Since the LR values are too small in this case, we use a log scale and the ratio between two instead of the distance to see the variation clearly. Note that monotonic increase of the value of $\lr\big(\mb R_\text{RCML}(r), \mb Z\big)$ w.r.t $r$ guarantees a unique optimal rank even if the optimization function as defined in (\ref{Eq:OptimalRank}) is not necessarily convex in $r$.

The algorithm to find the optimal rank is simple and not computationally expensive due to the analytical results above. For a given initial rank, we first determine a direction of searching and then find the optimal rank by increasing or decreasing the rank one by one. The value of the initial rank can be given by Brennan rule for the KASSPER data set and the number of jammers for a simulation model. The availability of the initial guess hastens the process of finding the optimal rank as shown in Algorithm 1.

\begin{algorithm}[t]
\caption{The proposed algorithm to select the rank via EL criterion}
\label{Alg:Rank}
    \begin{algorithmic}[1]
        \STATE Initialize the rank $r$ by physical environment such as Brennan rule.
        \STATE Evaluate $\lr(r-1)$, $\lr(r)$, $\lr(r+1))$, the LR values of RCML estimators for the ranks $r-1$, $r$, $r+1$, respectively.
            \begin{itemize}
                \item if $|\lr(r+1) - \lr_0| < |\lr(r) - \lr_0|$\\
                $\rightarrow$ increase $r$ by 1 until $|\lr(r) - \lr_0|$ is minimized to find $\hat r$.
                \item elseif $|\lr(r-1) - \lr_0| < |\lr(r) - \lr_0|$\\
                $\rightarrow$ decrease $r$ by 1 until $|\lr(r) - \lr_0|$ is minimized to find $\hat r$.
                \item else $\hat r = r$, the initial rank.
            \end{itemize}
    \end{algorithmic}
\end{algorithm}

\subsection{Joint selection of rank and noise power constraints}
\label{Sec:Both}

In this section, we investigate the second case that both the rank $r$ and the noise power $\sigma^2$ are not perfectly known. We propose the estimation of both the rank and the noise level based on the EL approach. The estimator with both the rank and the noise power obtained by the EL approach is given by
\be
\hat{\mb R}_{\text{RCML}_\text{EL}} = \hat \sigma^2 \mb{V} {\mb\Lambda^\star}^{-1}(\hat r) \mb{V}^H
\ee
where
\be
\label{Eq:OptimalRankSigma}
(\hat r, \hat \sigma^2) \equiv \arg\min_{r \in \mathds{Z}, \sigma^2 > 0} \Big| \lr\big(\mb R_\text{RCML}(r,\sigma^2), \mb Z\big)  - \lr_0 \Big|^2
\ee

In section \ref{Sec:Rankonly}, we have shown that the optimal rank via the EL approach is uniquely obtained for a fixed $\sigma^2$. Now we analyze the LR values of the RCML estimator for various $\sigma^2$ and a fixed rank.

\begin{lemma}
\label{Lemma2}
For a fixed rank, the LR value of the RCML estimator, which is a function of $\sigma^2$, has a maximum value at $\sigma^2 = \sigma_{\text{ML}}^2$. It monotonically increases for $\sigma^2 < \sigma_{\text{ML}}^2$ and monotonically decreases for $\sigma^2 > \sigma_{\text{ML}}^2$.
\end{lemma}
\begin{proof}
We first represent the LR values as a function of $\sigma^2$ and show the function is increasing or decreasing according to the sign of the first derivative. See Appendix B for details.
\end{proof}

Fig. \ref{Fig:LR_sigma} shows an example of the LR values as a function of the noise level $\sigma^2$ when two optimal solutions exist. As shown in Lemma \ref{Lemma2}, we see that the LR value is maximized for the ML solution of $\sigma^2$. It is obvious that we have three cases of the number of the solution of the optimal noise power for given a fixed rank from Lemma \ref{Lemma2}: 1) no solution if $\lr_0 > \lr(\sigma^2_{\text{ML}})$, 2) only one solution if $\lr_0 = \lr(\sigma^2_{\text{ML}})$, and 3) two optimal solutions if $\lr_0 < \lr(\sigma^2_{\text{ML}})$. Now we discuss how to obtain the optimal noise power for a fixed rank.

\begin{lemma}
\label{Lemma3}
For given a fixed rank, $r$, satisfying $\lr_0 < \lr(r,\sigma^2_{\text{ML}})$ , the noise power obtained by the expected likelihood approach, $\hat{\sigma}_{\text{EL}}^2$, is given by
\be
\hat{\sigma}_{\text{EL}}^2 = \exp\Bigg(W_k\bigg(\frac{b}{a} e^{-\frac{c}{a}}\bigg) + \frac{c}{a}\Bigg)
\ee
where $W_k (z)$ is the $k$-th branch of Lambert $W$ function, $k=0,1$, and
\be
\left\{\begin{array}{l}
a = r - N\\
b = \sum_{k=r+1}^N d_k\\
c = \log \lr_0 - \log\Big (\prod_{k=r+1}^N d_k\Big) + a
\end{array} \right.
\ee
\end{lemma}
\begin{proof}
We first set $\lr(\sigma^2)$ to $\lr_0$ and rewrite the equation by using a transformation of variables. The equation is reduced to a well-known form whose solution is expressed by a Lambert $W$ function. See Appendix C for details.
\end{proof}
Lemma \ref{Lemma3} shows that there is a closed-form solution of the optimal noise power for a fixed rank. Therefore we do not need expensive iterative or numerical algorithms to find the optimal noise power.

Now we propose the method to alternately find the optimal solution of both the rank and the noise power. For a fixed $\sigma^2$, we can obtain the optimal rank via Algorithm 1. For a fixed rank, we should consider three cases described above. For the first case that the LR value corresponding $\sigma_{\text{ML}}^2$ is less than $\lr_0$, we increase the rank until at least one of the solutions of $\sigma^2$ exists. For the second case, we can easily determine $\hat \sigma^2 = \sigma_{\text{ML}}^2$. For the third case that there are two solutions of $\sigma^2$, we have to choose one among two EL solutions and the ML solution. We intuitively observe that with target-free training samples the values of the test statistics such as the normalized matched filter given in \eqref{Eq:NMF} are typically smaller for the better estimator since smaller values of the test statistics clearly separate the values from observations including target information and lead to higher detection probability. Therefore, we generate the values of the test statistics for estimates with $\sigma_{\text{ML}}^2$, $\sigma_{\text{EL1}}^2$, $\sigma_{\text{EL2}}^2$ and choose one that generates the smallest average value of the test statistics. The detailed procedure of jointly determining the best rank and noise power is described in Algorithm 2.

\be
\label{Eq:NMF}
\dfrac{|\mb{s}^H \hat{\mb{R}}^{-1}\mb{z}|^2}{(\mb{s}^H \hat{\mb{R}}^{-1}\mb{s})(\mb{z}^H \hat{\mb{R}}^{-1}\mb{z})} \overset{H_1}{\underset{H_0}{\gtrless}} \lambda_{\text{NMF}}
\ee

\begin{figure}
\centering
\includegraphics[scale=0.5]{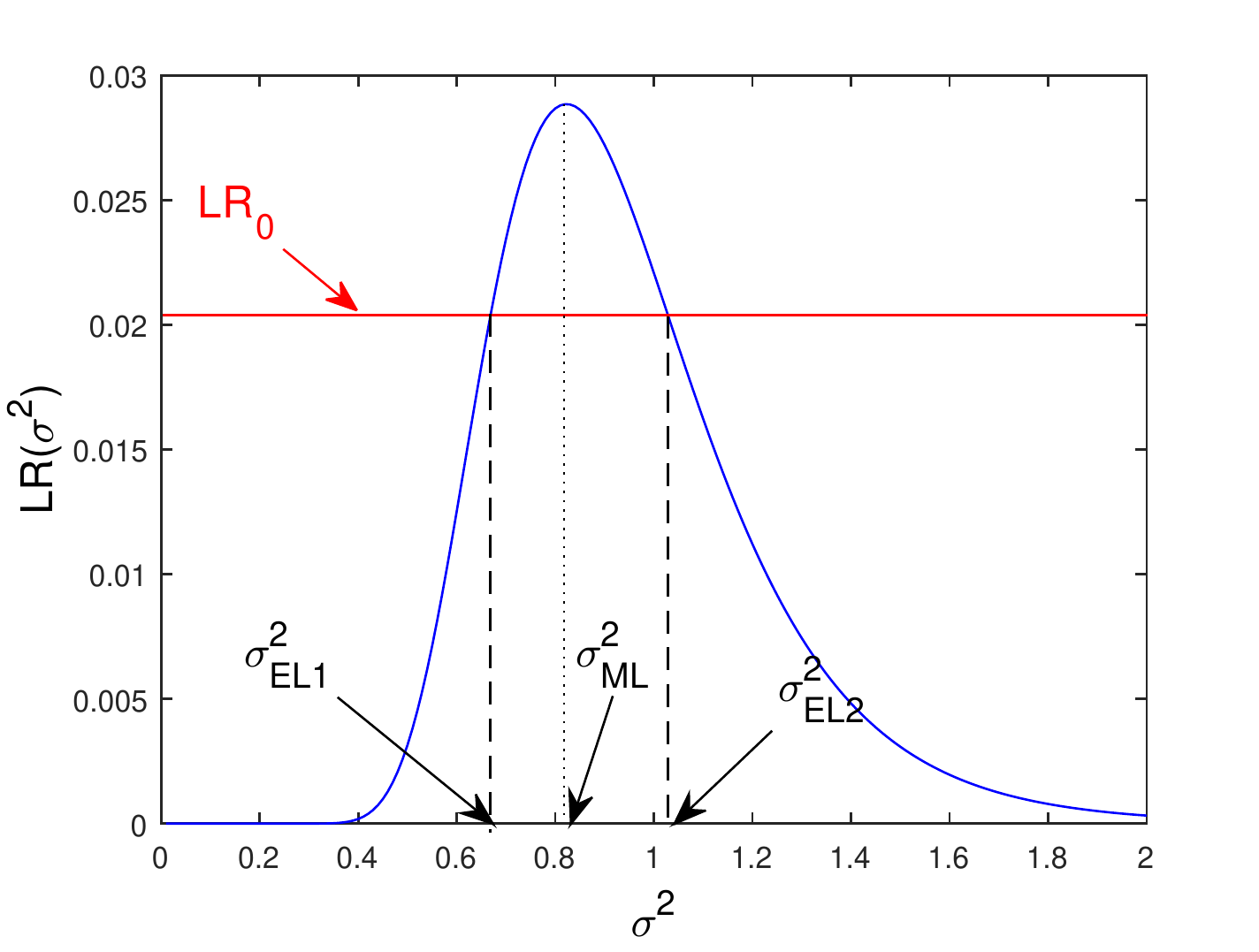}
\caption{The LR value versus $\sigma^2$ for the simulation model, $N=20$, $K=40$, $r=5$}
\label{Fig:LR_sigma}
\end{figure}

\begin{algorithm}[t]
\caption{The proposed algorithm to select the rank and the noise level via EL}
\label{Alg:RankNoise}
    \begin{algorithmic}[1]
        \STATE Initialize the rank $r$ by physical environment such as Brennan rule or the number of jammers.
        \STATE If there is no solution of $\sigma^2$ for given $r$, increase $r$ until the solution of $\sigma^2$ exists.
        \STATE Obtain $\sigma_{\text{ML}}^2 =\frac{1}{N-r} \sum_{i=r+1}^N d_i$.
        \STATE For given $\sigma_{\text{ML}}^2$, find a new $r$ using Algorithm 1.
        \STATE Repeat Step 3 and Step 4 until the rank $r$ converges.
        \STATE After $r$ is determined, choose $\hat \sigma^2$ among $\sigma_{\text{ML}}^2$, $\sigma_{\text{EL1}}^2$, $\sigma_{\text{EL2}}^2$.
    \end{algorithmic}
\end{algorithm}

\subsection{Selection of condition number constraint}
\label{Sec:ConditionNumber}

Now we propose a method to determine the condition number constraint through the EL approach in this section. As shown in Eq. \eqref{Eq:CN1} through Eq. \eqref{Eq:Giu2}, the condition number constrained ML estimator is a function of $u$ which is a function of the condition number $K_{\max}$. Therefore, the final estimate is also a function of $K_{\max}$. Our goal is to find the optimal condition number so that the LR value of the estimated covariance matrix should be as close as possible to the statistical median value of the LR value of the true covariance matrix, that is
\be
\hat{\mb R}_{\text{CNCML}_\text{EL}} = \hat \sigma^2 \mb{V} {\mb\Lambda^\star}^{-1}(\hat K_{\max}) \mb{V}^H
\ee
where
\be
\label{Eq:OptimalCN}
\hat K_{\max} \equiv \arg\min_{K_{\max} \geq 1} \Big| \lr\big(\mb R_\text{CNCML}(K_{\max}), \mb Z\big)  - \lr_0 \Big|^2
\ee

Before we discuss the algorithm to find the optimal condition number, we analyze the closed-form solution for the condition number constrained ML estimation which is proposed in \cite{Aubry12}. We derive a more explicit closed-form solution to analyze the LR values of the estimator more tractably.

\begin{lemma}
\label{Lemma4}
The more simplified closed-form solution of the condition number constrained ML estimator is given by
   \begin{enumerate}
    \item $d_1\leq \sigma^2$,
    \be
    \hat{\mb R}_{\text{CN}} = \sigma^2 \mb I
    \ee

    \item $\sigma^2 \leq d_1 \leq \sigma^2 K_{\max}$,
    \be
    \hat{\mb R}_{\text{CN}} = \hat{\mb R}_{\text{FML}}
    \ee

    \item $d_1 > \sigma^2 K_{\max}$ and $K_{\max} \geq \frac{\sum_{i=1}^c d_i}{c - \sum_{\bar N + 1}^N (d_i -1)}$,
    \be
    \hat{\mb R}_{\text{CN}} = \mb\Phi \diag(\bs\lambda^*) \mb\Phi^H
    \ee
    where
    \be
    \bs\lambda^\star = \big[ \sigma^2 K_{\max}, \ldots, \sigma^2 K_{\max}, d_{c+1}, \ldots, d_{\bar N}, \sigma^2,\ldots,\sigma^2  \big],
    \ee
     $c$ and $\bar N$ are the vector of the eigenvalues of the estimate, the largest indices so that $d_c > \sigma^2 K_{\max}$, and $d_{\bar N} \geq \sigma^2$

    \item $d_1 > \sigma^2 K_{\max}$ and $K_{\max} < \frac{\sum_{i=1}^c d_i}{c - \sum_{\bar N + 1}^N (d_i -1)}$,
    \be
    \bs\lambda^\star = \big[ \frac{\sigma^2}{u}, \ldots, \frac{\sigma^2}{u}, d_{p+1}, \ldots, d_q,\frac{\sigma^2}{uK_{\max}},\ldots,\frac{\sigma^2}{uK_{\max}}  \big]
    \ee
    \end{enumerate}
and the condition numbers of the estimates are $1$, $\frac{d_1}{\sigma^2}$, $K_{\max}$, and $K_{\max}$, respectively.
\end{lemma}
\begin{proof}
In each case, we derive the closed form using $\bar u$ which is the optimal solution of \eqref{Eq:OptimizationU} provided in \cite{Aubry12}. See Appendix D for details.
\end{proof}

From Lemma \ref{Lemma4}, for the first two cases that is $d_1 \leq \sigma^2 K_{\max}$, the estimator is either a scaled identity matrix or the FML. Therefore, there is no need to find an optimal condition number in these cases since the estimator is not a function of the condition number.

Now we investigate uniqueness of the optimal condition number as we have done in the case of only rank constraint for the last two cases where the optimal eigenvalues are functions of the condition number.

\begin{lemma}
\label{Lemma5}
The LR value of the condition number ML estimator is a monotonically increasing function with respect to the condition number $K_{\max}$ and there is only one unique $K_{{\max}_{\text{EL}}}$.
\end{lemma}
\begin{proof}
We simplify $\lr(K_{\max})$ and evaluate the first derivative. Then we show its increasing property in each case in Lemma \ref{Lemma4}. See Appendix E for details.
\end{proof}

Lemma \ref{Lemma5} formally proves that the there exist only one optimal condition number and therefore we can find the optimal condition number numerically. The algorithm of finding the global optimal condition number is shown in Algorithm \ref{Alg:ConditionNumber}. We first set the initial condition number as the ML condition number obtained by \cite{Aubry12}. Then we increase or decrease the condition number to the direction where the LR value decreases. Reducing the stepsize as the direction is reversed, we find the optimal condition number as precisely as we want.

\begin{algorithm}[t]
\caption{The proposed algorithm to select condition number via EL }
\label{Alg:ConditionNumber}
    \begin{algorithmic}[1]
        \STATE Obtain the ML solution of the condition number $K_{\max_{\text{ML}}}$ by the method in \cite{Aubry12} and set the initial value of $K_{\max} = K_{\max_{\text{ML}}}$
        \STATE Set the initial step, $\Delta = K_{\max}/100$
        \STATE Evaluate $\lr(K_{\max}-\Delta)$, $\lr(K_{\max})$, $\lr(K_{\max}+\Delta)$
            \begin{itemize}
                \item if $|\lr(K_{\max_{\text{ML}}}+\Delta) - \lr_0| < |\lr(K_{\max_{\text{ML}}}) - \lr_0|$\\
                $\rightarrow$ increase $K_{\max}$ by $\Delta$ until it does not hold.\\
                $\rightarrow$ then $\Delta = -\Delta/10$
                \item elseif $|\lr(K_{\max_{\text{ML}}}+\Delta) - \lr_0| > |\lr(K_{\max_{\text{ML}}}) - \lr_0|$\\
                $\rightarrow$ decrease $K_{\max}$ by $\Delta$ until it does not hold.\\
                $\rightarrow$ then $\Delta = -\Delta/10$
            \end{itemize}
         \STATE Repeat Step 3 until $\Delta < 0.0001$.
    \end{algorithmic}
\end{algorithm}

\section{Experimental Validation}
\label{Sec:Experiments}

\subsection{Experimental setup}

We focus on structured covariance estimation techniques which incorporate rank, noise power and condition number constraints. Two data sets are used in the experiments: 1) a radar covariance simulation model and 2) the KASSPER dataset \cite{Bergin02}.

First, we consider a radar system with an $N$-element uniform linear array for the simulation model. The overall covariance which is composed of jammer and additive white noise can be modeled by
\be
\label{Eq:SimulationModel}
\mb{R}(n,m) = \sum_{i=1}^J \sigma_i^2 \sinc[0.5 \beta_i (n-m) \phi_i ] e^{j(n-m)\phi_i} + \sigma_a^2 \delta(n,m)
\ee
where $n,m \in \{1,\ldots,N\}$, $J$ is the number of jammers, $\sigma_i^2$ is the power associated with the $i$th jammer, $\phi_i$ is the jammer phase angle with respect to the antenna phase center, $\beta_i$ is the fractional bandwidth, $\sigma_a^2$ is the actual power level of the white disturbance term, and $\delta(n,m)$ has the value of 1 only when $n=m$ and 0 otherwise. This simulation model has been widely and very successfully used in previous literature \cite{Steiner00,Aubry12,Pallotta12,DeMaio09} for performance analysis.

Data from the L-band data set of KASSPER program is the other data set used in our experiments. Note that the KASSPER data set exhibits two desirable characteristics: 1) the low-rank structure of clutter and 2) the true covariance matrices for each range bin have been made available. These two characteristics facilitate comparisons via powerful figures of merit. The L-band data set consists of a data cube of 1000 range bins corresponding to the returns from a single coherent processing interval from $11$ channels and $32$ pulses. Therefore, the dimension of observations (spatio-temporal product) $N$ is $11 \times 32 = 352$. Other parameters are detailed in Table \ref{Tb:parameters}.

\begin{table}[!t]
\begin{center}
\caption{KASSPER Dataset-1 parameters}
\label{Tb:parameters}
\begin{tabular}{l l}
  \hline
  Parameter & Value\\
  \hline
  Carrier Frequency & 1240 MHz \\
  Bandwidth (BW) & 10 MHz \\
  Number of Antenna Elements & 11 \\
  Number of Pulses & 32 \\
  Pulse Repetition Frequency & 1984 Hz \\
  1000 Range Bins & 35 km to 50 km \\
  91 Azimuth Angles & $87^{\circ}$, $89^{\circ}$, $\ldots$ $267^{\circ}$ \\
  128 Doppler Frequencies & -992 Hz, -976.38 Hz, $\ldots$, 992 Hz \\
  Clutter Power & 40 dB \\
  Number of Targets & 226 (~200 detectable targets) \\
  Range of Target Dop. Freq. & -99.2 Hz to 372 Hz\\
  \hline
\end{tabular}
\end{center}
\end{table}

As a figure of merit, we use the normalized signal to interference and noise ratio (SINR). The normalized SINR measure is widely used and given by
\be
\eta = \dfrac{|\mb{s}^H\hat{\mb{R}}^{-1}\mb{s}|^2}{|\mb{s}^H\hat{\mb{R}}^{-1}\mb{R}\hat{\mb{R}}^{-1}\mb{s}||\mb{s}^H\mb{R}^{-1}\mb{s}|}
\ee
where $\mb s$ is the spatio-temporal steering vector, $\hat{\mb R}$ is the data-dependent estimate of $\mb R$, and $\mb R$ is the true covariance matrix. It is easily seen that $0<\eta<1$ and $\eta=1$ if and only if $\hat{\mb R} = \mb R$. The SINR is plotted in decibels in all our experiments, that is, $\text{SINR} \text{(dB)} = 10 \log_{10} \eta$. Therefore, $\text{SINR} \text{(dB)} \leq 0$. For the KASSPER data set, since the steering vector is a function of both azimuthal angle and Doppler frequency, we obtain plots as a function of one variable (azimuthal angle or Doppler) by marginalizing over the other variable. We evaluate and compare different covariance estimation techniques and parameter selection methodsas given by:
\begin{itemize}
    \item \textbf{Sample Covariance Matrix:} The sample covariance matrix is given by $\mb S = \frac{1}{K} \mb Z \mb Z^H$. It is well known that $\mb S$ is the unconstrained ML estimator under Gaussian disturbance statistics. We refer to this as SMI.
    \item \textbf{Fast Maximum Likelihood:} The fast maximum likelihood (FML) \cite{Steiner00} uses the structural constraint of the covariance matrix. The FML method just involves the eigenvalue decomposition of the sample covariance and perturbing eigenvalues to conform to the structure. The FML also can be considered as the RCML estimator with the rank which is the greatest index $i$ satisfying $\lambda_i > \sigma^2$ where $\lambda_i$'s are the eigenvalues of the sample covariance in descending order. Therefore, a rank can be considered as an output of the FML. The FML's success in radar STAP is widely known \cite{Rangaswamy04Sep}.
    \item \textbf{Rank Constrained ML Estimators:} The RCML estimator with the rank or the rank and the noise level obtained by the proposed methods using the expected likelihood approach. The rank is obtained by the EL approach in the case of the imperfect rank constraint and both of the rank and the noise level are obtained by the EL approach in the case of imperfect rank and noise power constraints. We refer to these as \RCMLEL \hspace{-1mm}.
    \item \textbf{Chen \emph{et al.} Rank Selection Method:} Chen \emph{et al.} \cite{Chen01} proposed a statistical procedure for detecting the multiplicity of the smallest eigenvalue of the structured covariance matrix using statistical selection theory. The rank can be estimated from their methods using pre-calculated parameters. We refer to this method as \RCMLC \hspace{-1mm}.
    \item \textbf{AIC:} Akaike \cite{Akaike74} proposed the information theoretic criteria for model selection. The Akaike's imformation criteria (AIC) selects the model that best fits the data for given a set of observations and a family of models, that is, a parameterized family of probability densities. Wax and Kailath \cite{Wax85} proposed the method to determine the number of signals from the observed data based on the AIC. Since their method only determines the rank, we compare the RCML estimator with the rank obtained by their method. We refer to this method as RCML$_{\text{AIC}}$.
    \item \textbf{Condition number constrained ML estimators:} The maximum likelihood estimation method of the covariance matrix with a condition number \cite{Aubry12} proposed by Aubry \emph{et al.} is considered for evaluating the performance with three different condition numbers. 1) \CNCML\hspace{-1mm}: the condition number obtained by the proposed method in \cite{Aubry12}, and 2) \CNCEL\hspace{-1mm}: the condition number obtained by the expected likelihood approach.
\end{itemize}

\subsection{Rank constraint}
\label{Sec:ResultRank}

\begin{figure}
\centering
\includegraphics[scale=0.5]{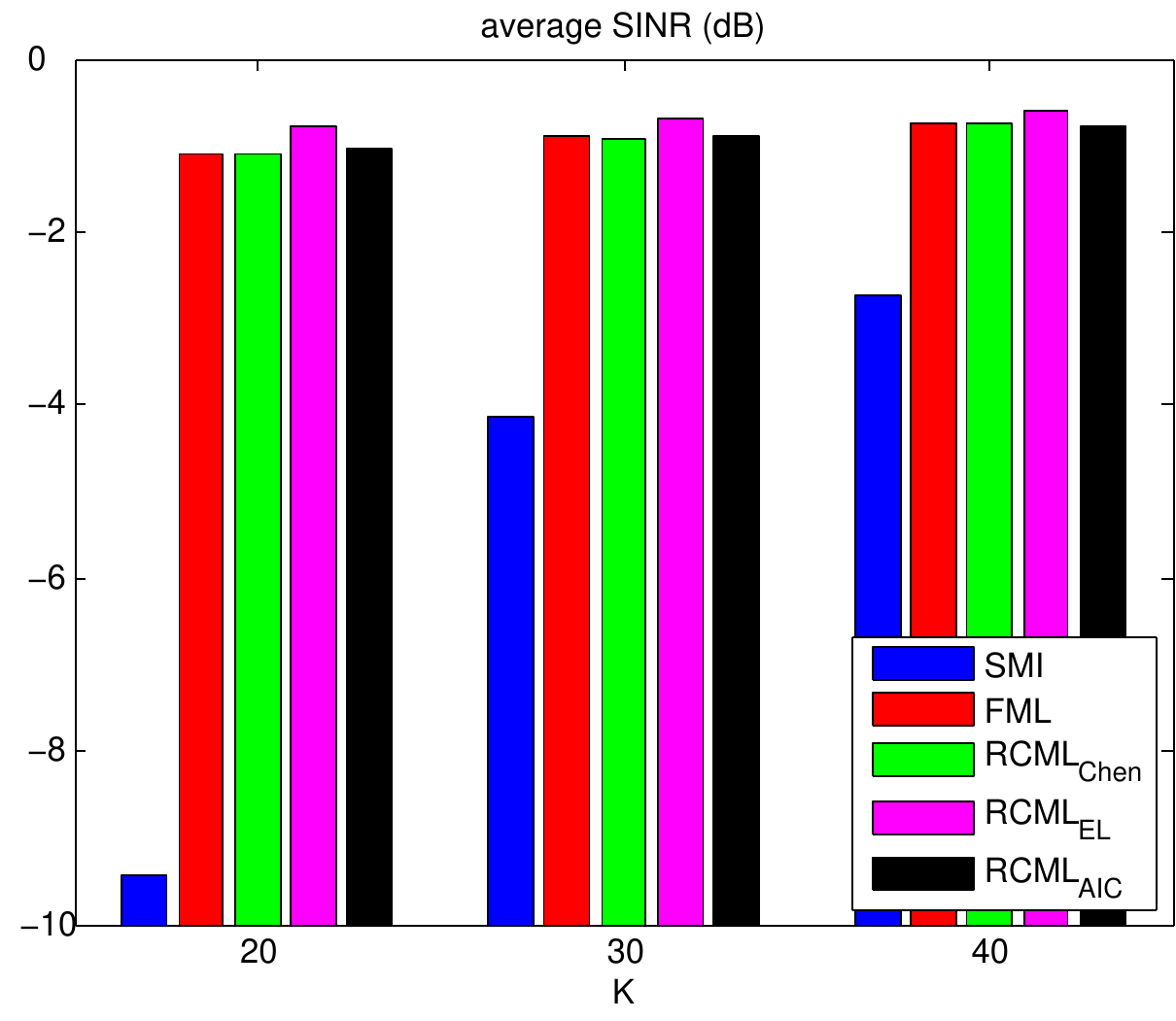}\label{Fig:SINR_Simulation}
\caption{Normalized SINR in dB versus number of training samples $K$ $(N=20)$ for the simulation model.}
\label{Fig:Simulation_rank}
\end{figure}

\begin{table}[!t]
\begin{center}
\caption{Ranges of ranks estimated by various methods}
\label{Tb:RankEstimated}
\begin{tabular}{|c|c|c|c|c|c|}
  \hline
  & SMI & FML & \RCMLC & \RCMLEL & \RCMLAIC\\
  \hline
  Simulation & 20 & 11-13 & 18-20 & 3-5 & 4-7\\
  \hline
  KASSPER & 352 & 200-210 & 300-350 & 41-45 & 47-60\\
  \hline
  Corrupted & 352 & 200-210 & 300-350 & 41-45 & 47-70\\
  \hline
\end{tabular}
\end{center}
\end{table}

\begin{figure*}[!t]
\begin{center}
\subfloat[]{\includegraphics[scale=0.5]{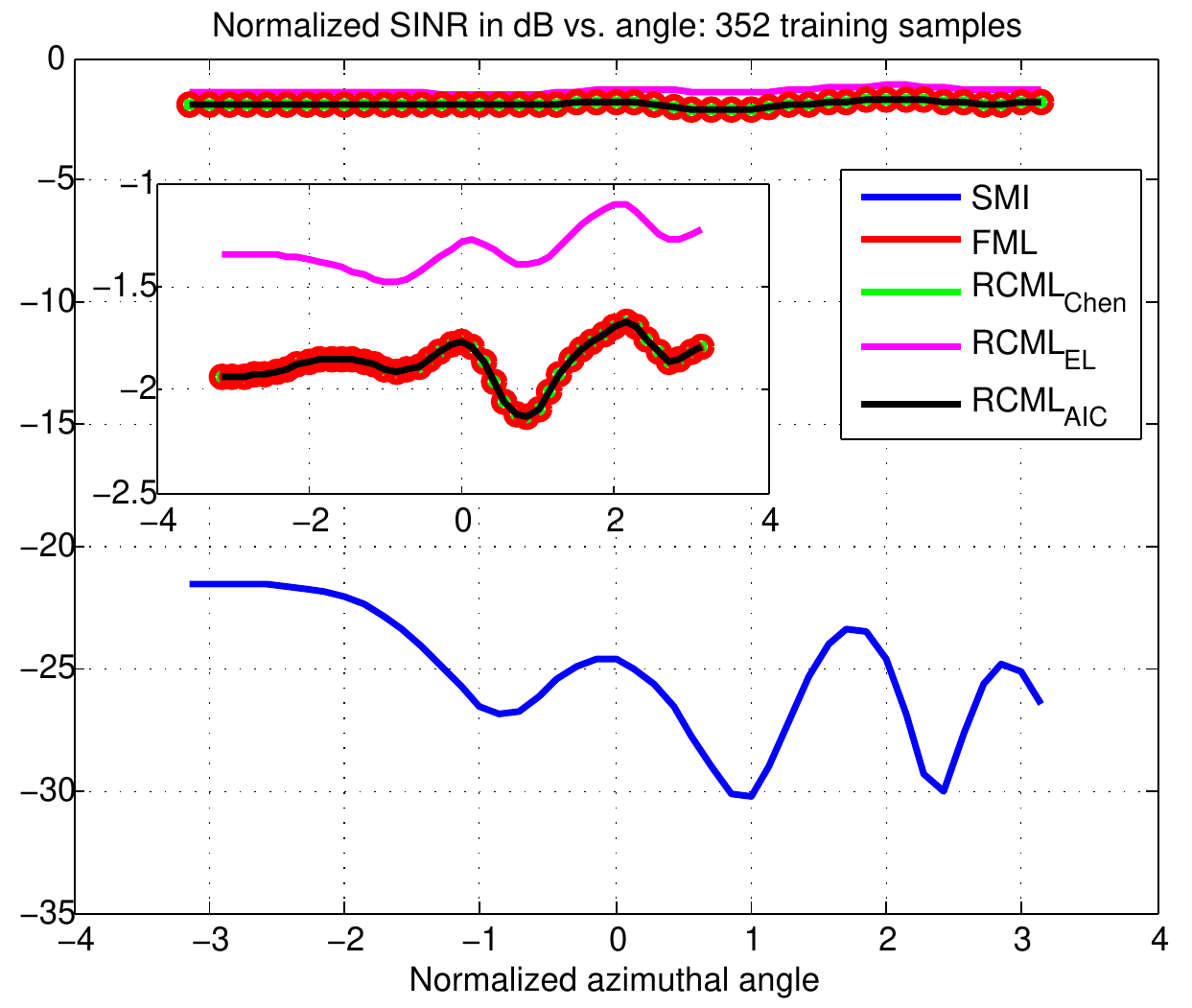}\label{Fig:KASSPER_rank_angle_352}}
\hfil
\subfloat[]{\includegraphics[scale=0.5]{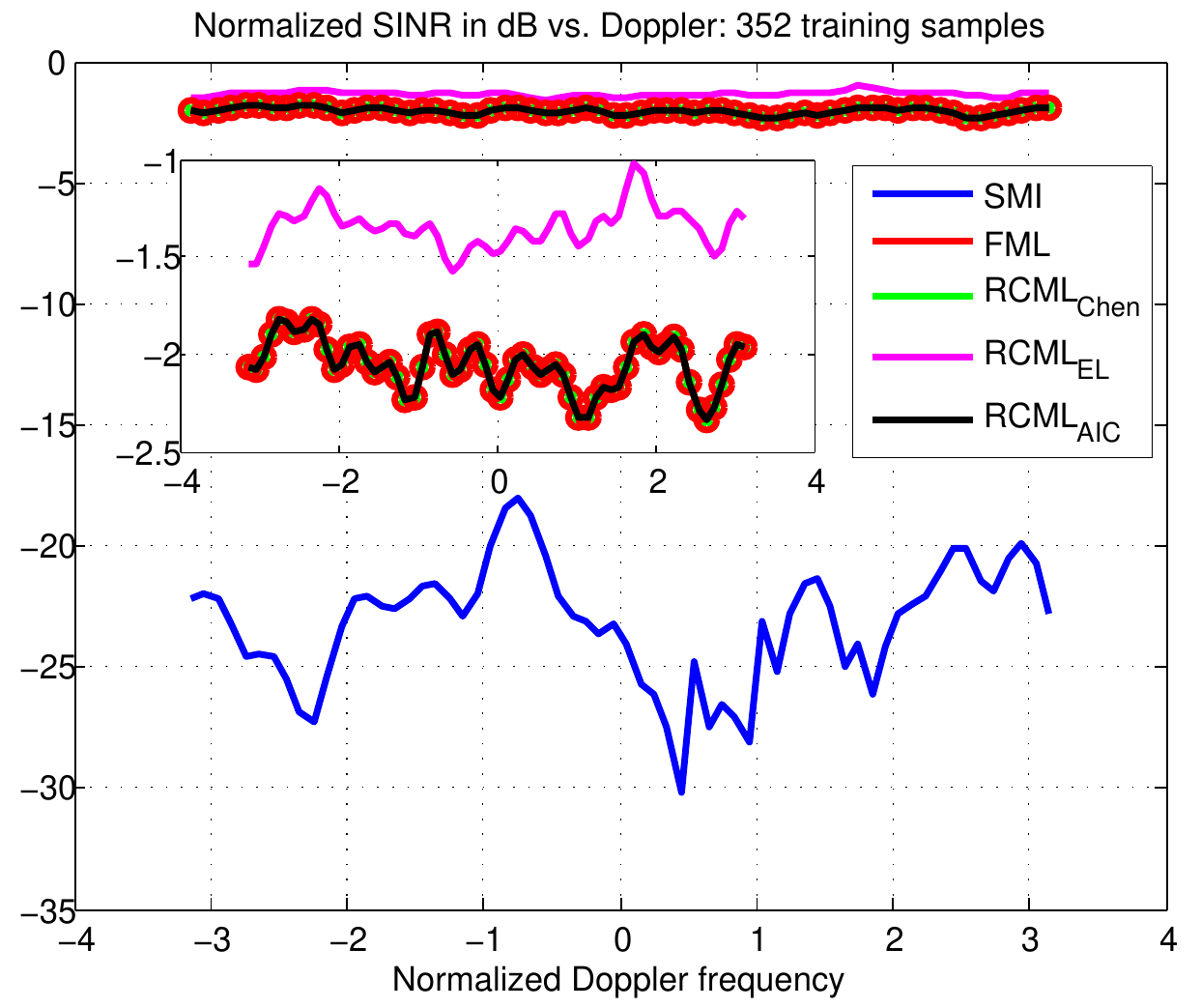}\label{Fig:KASSPER_rank_dop_352}}\\
\subfloat[]{\includegraphics[scale=0.5]{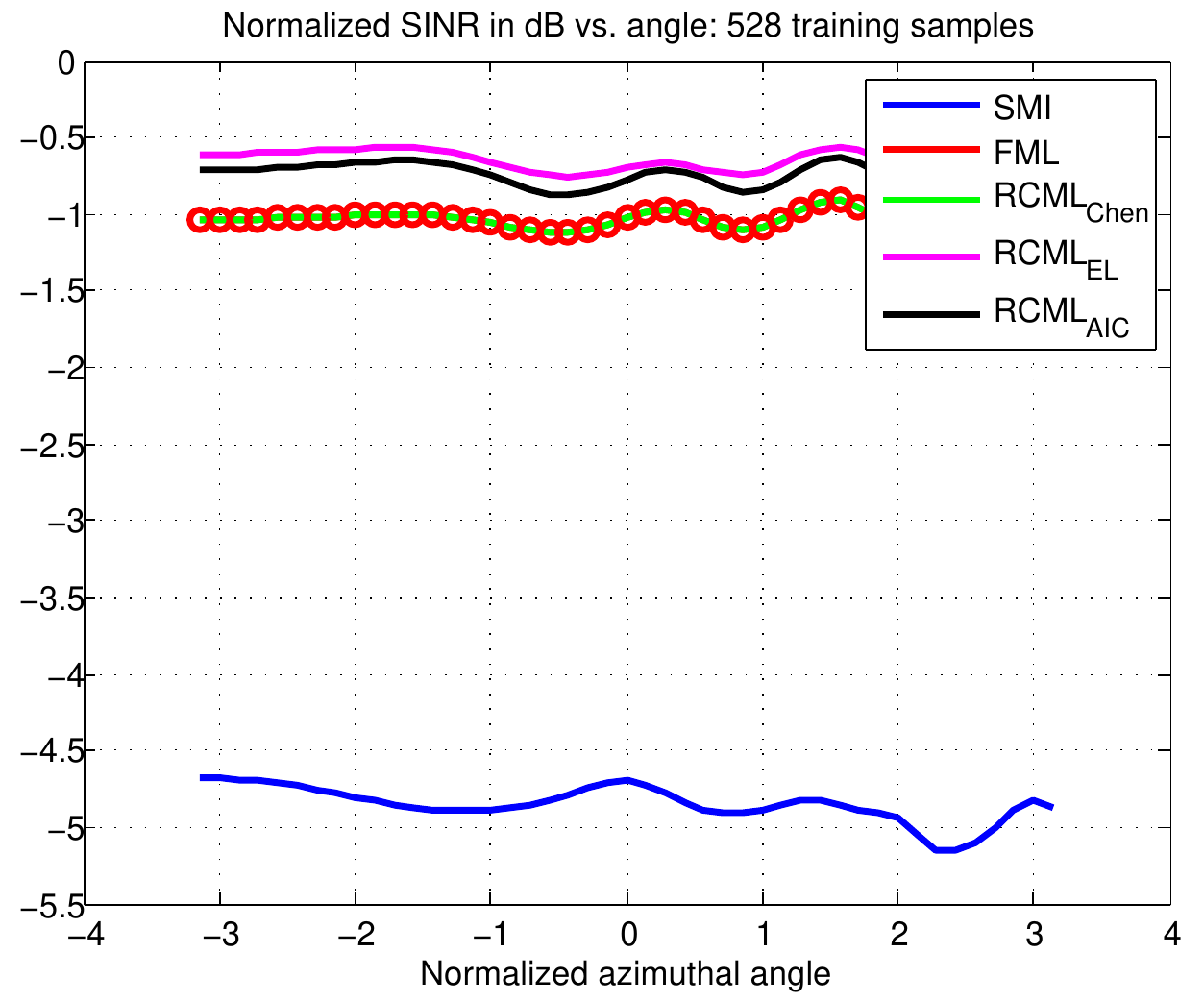}\label{Fig:KASSPER_rank_angle_528}}
\hfil
\subfloat[]{\includegraphics[scale=0.5]{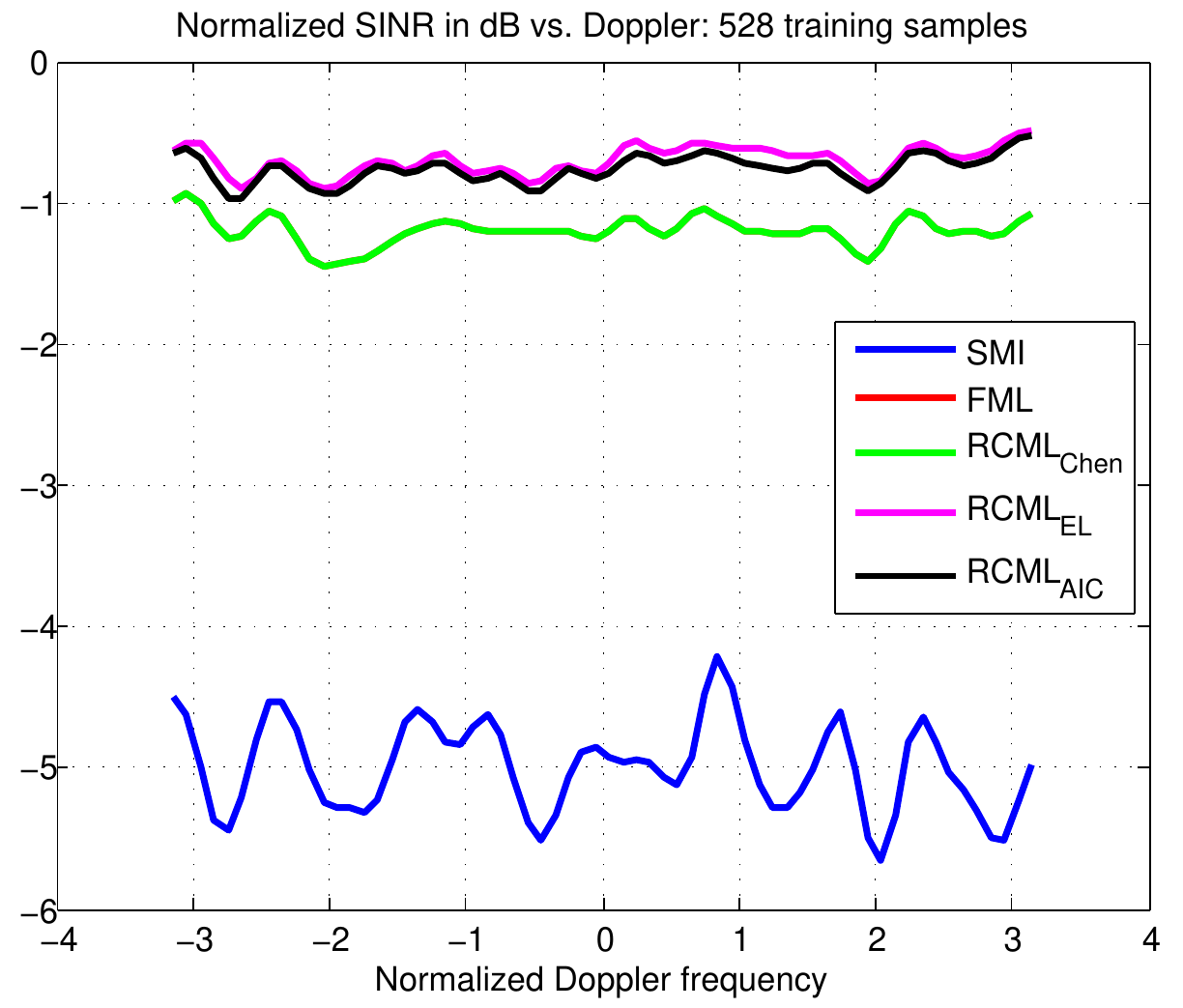}\label{Fig:KASSPER_rank_dop_528}}\\
\subfloat[]{\includegraphics[scale=0.5]{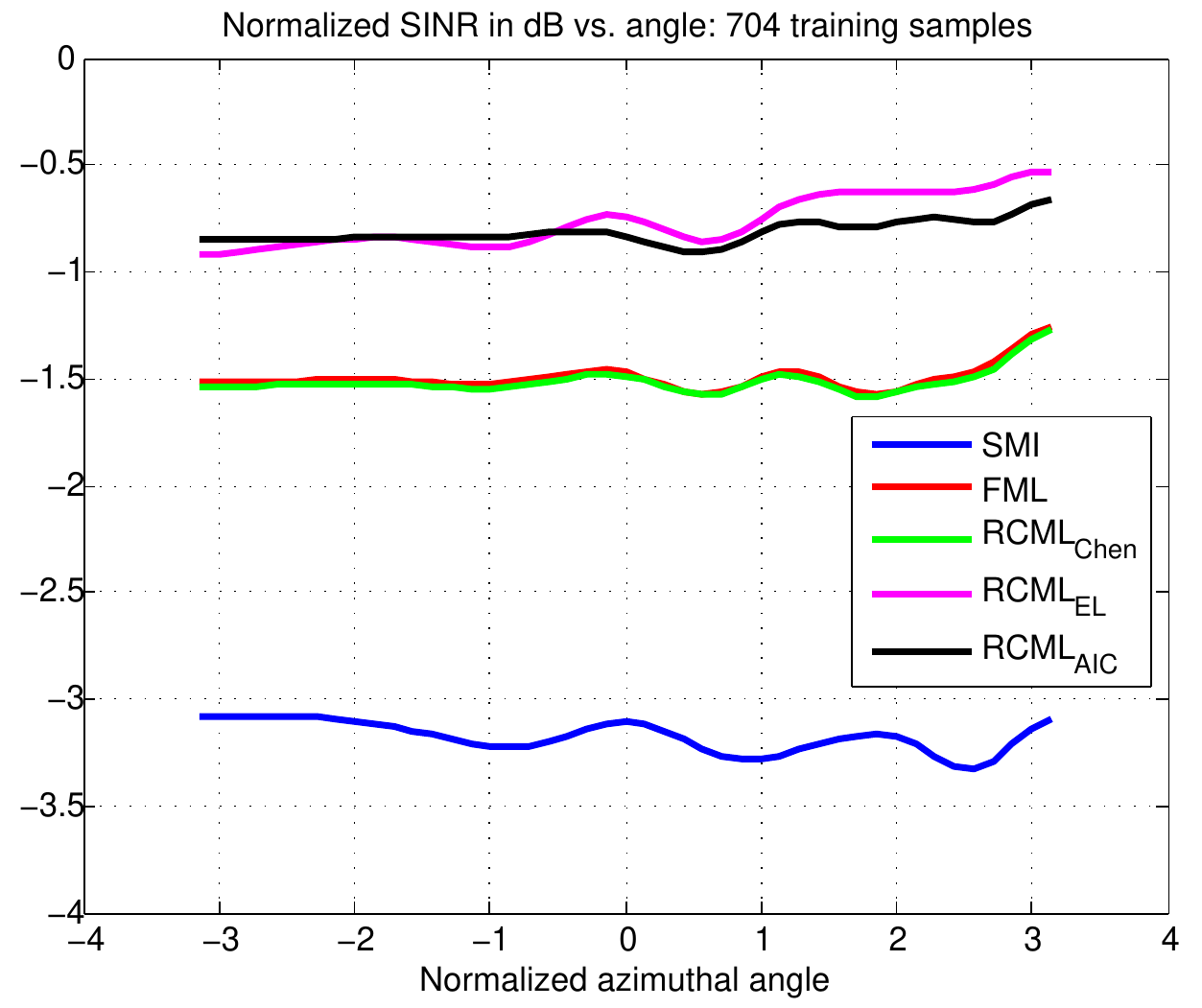}\label{Fig:KASSPER_rank_angle_704}}
\hfil
\subfloat[]{\includegraphics[scale=0.5]{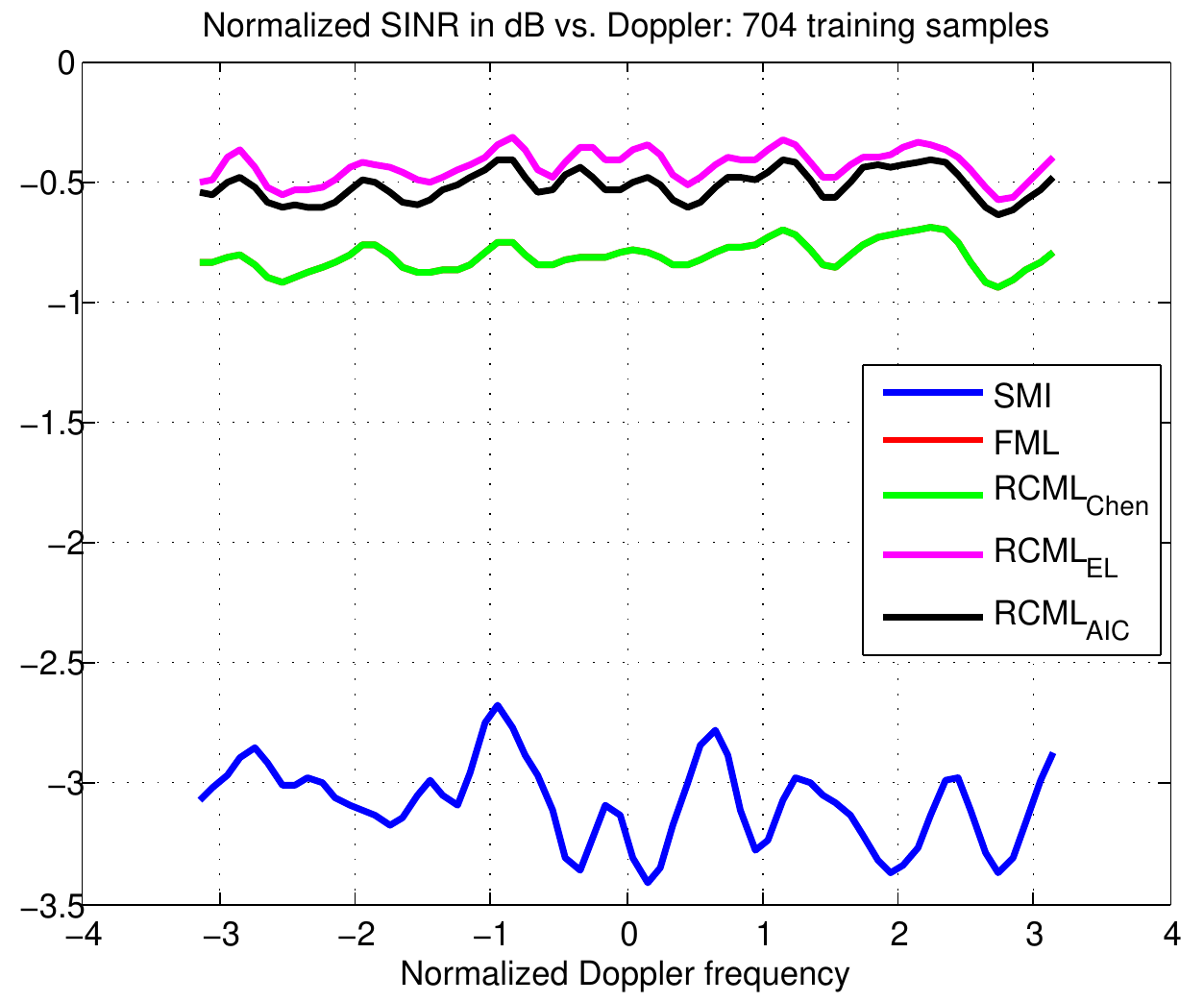}\label{Fig:KASSPER_rank_dop_704}}
\end{center}
\caption{Normalized SINR versus azimuthal angle and Doppler frequency for the KASSPER data set. (a) and (b) for $K=N=352$, (c) and (d) for $K=1.5N=528$, and (e) and (f) for $K=2N=704$.}
\label{Fig:KASSPER_rank}
\end{figure*}

\begin{figure*}[!t]
\begin{center}
\subfloat[]{\includegraphics[scale=0.5]{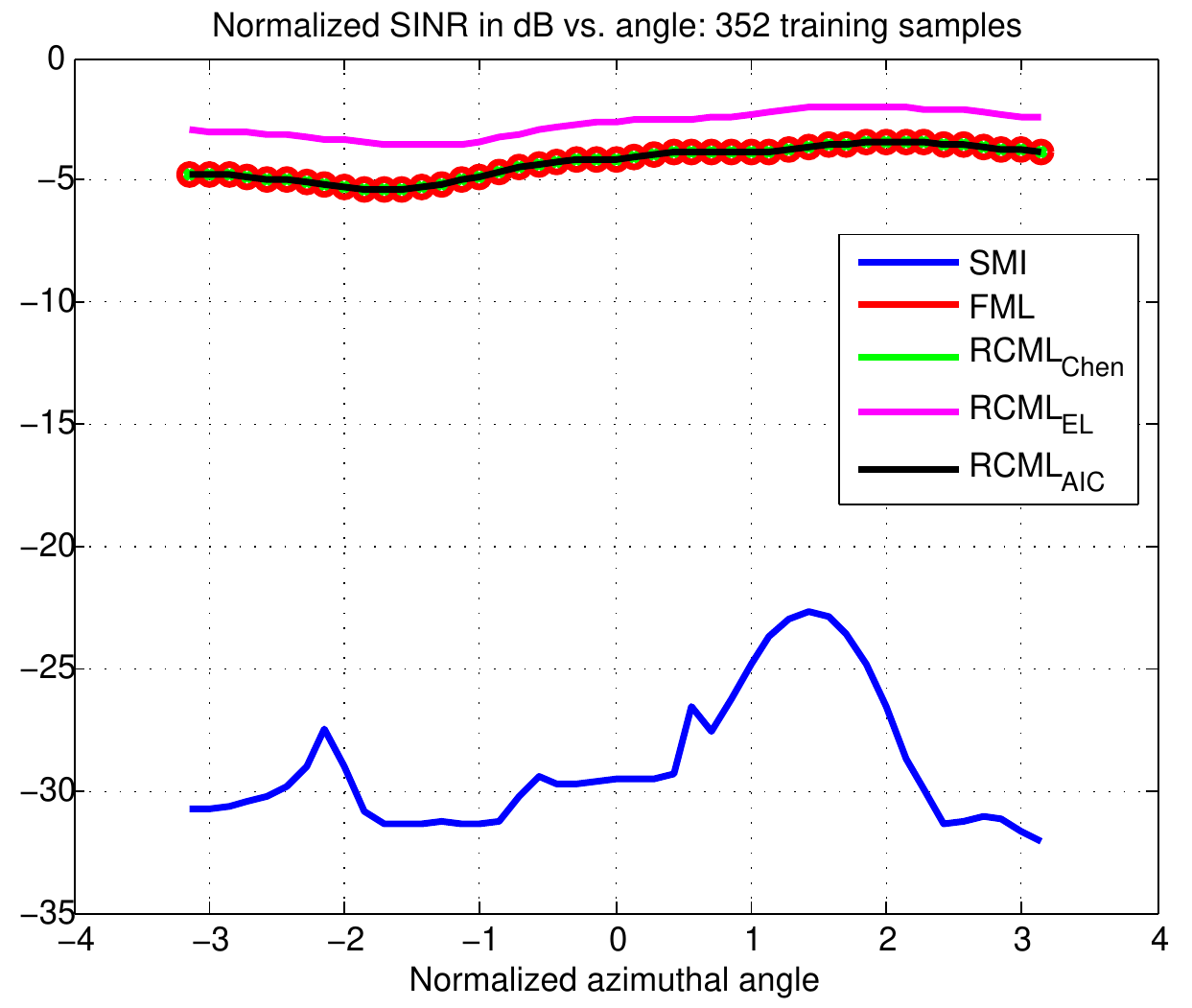}\label{Fig:KASSPER_angle_nonhomo_352}}
\hfil
\subfloat[]{\includegraphics[scale=0.5]{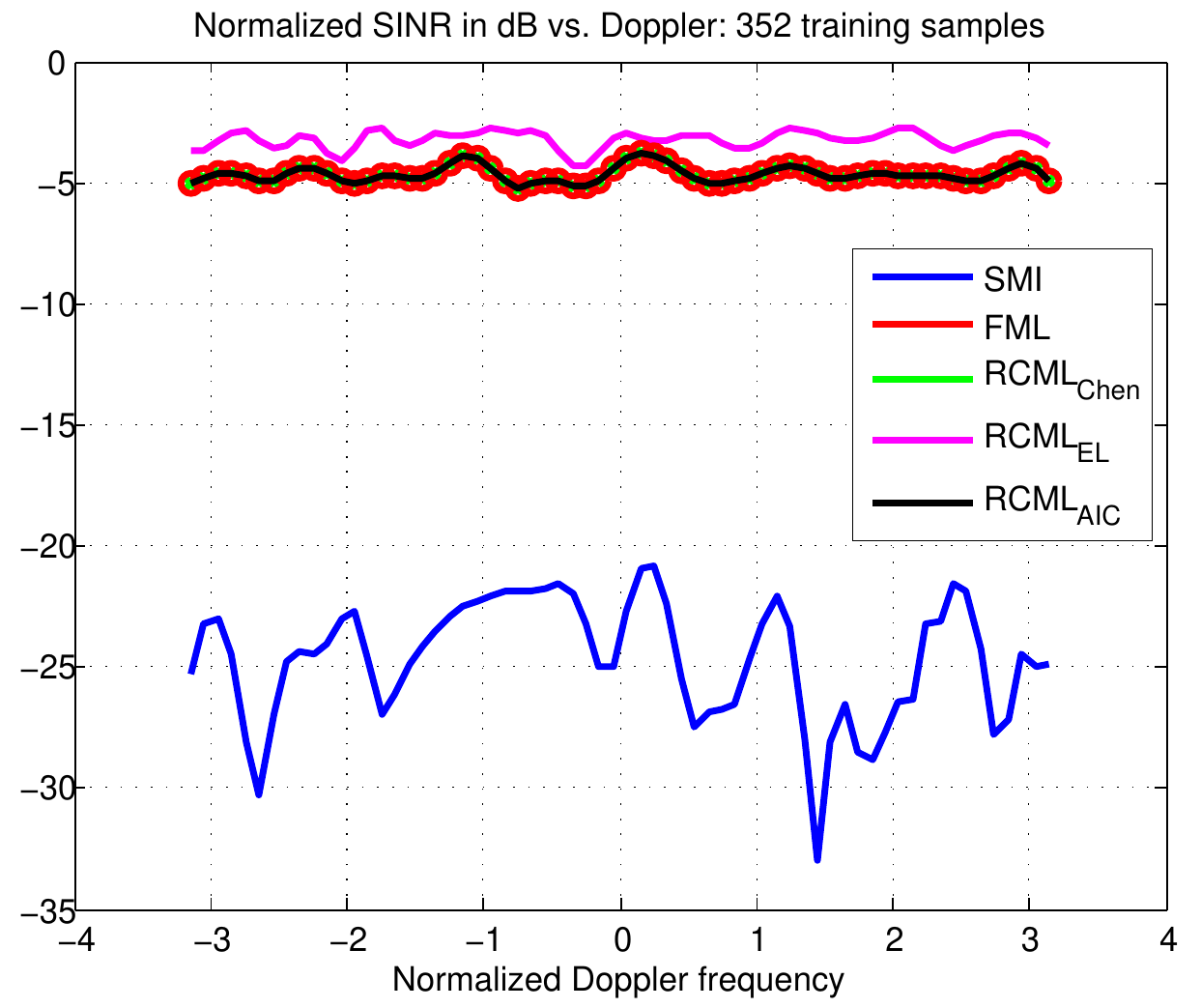}\label{Fig:KASSPER_dop_nonhomo_352}}\\
\subfloat[]{\includegraphics[scale=0.5]{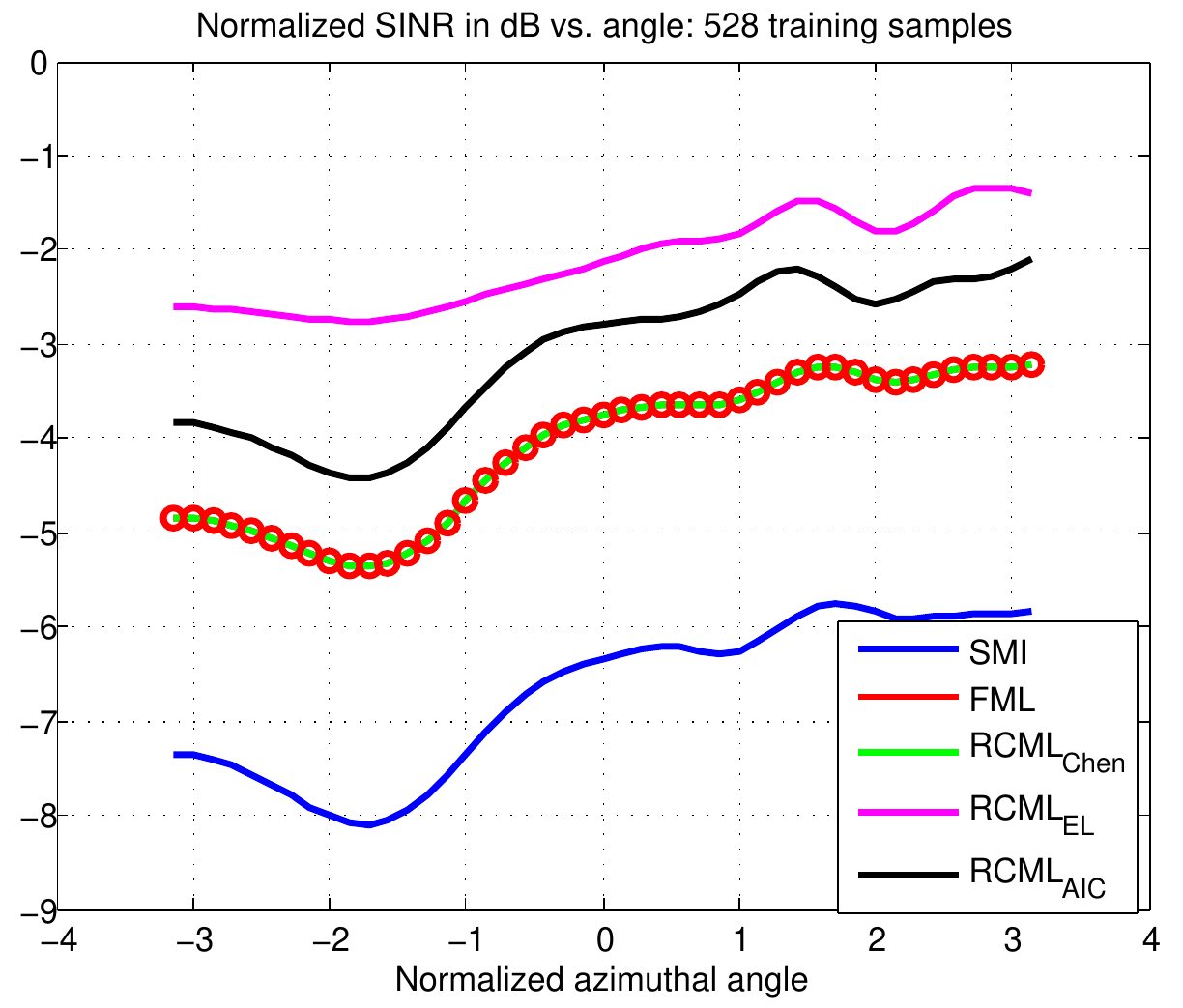}\label{Fig:KASSPER_angle_nonhomo_528}}
\hfil
\subfloat[]{\includegraphics[scale=0.5]{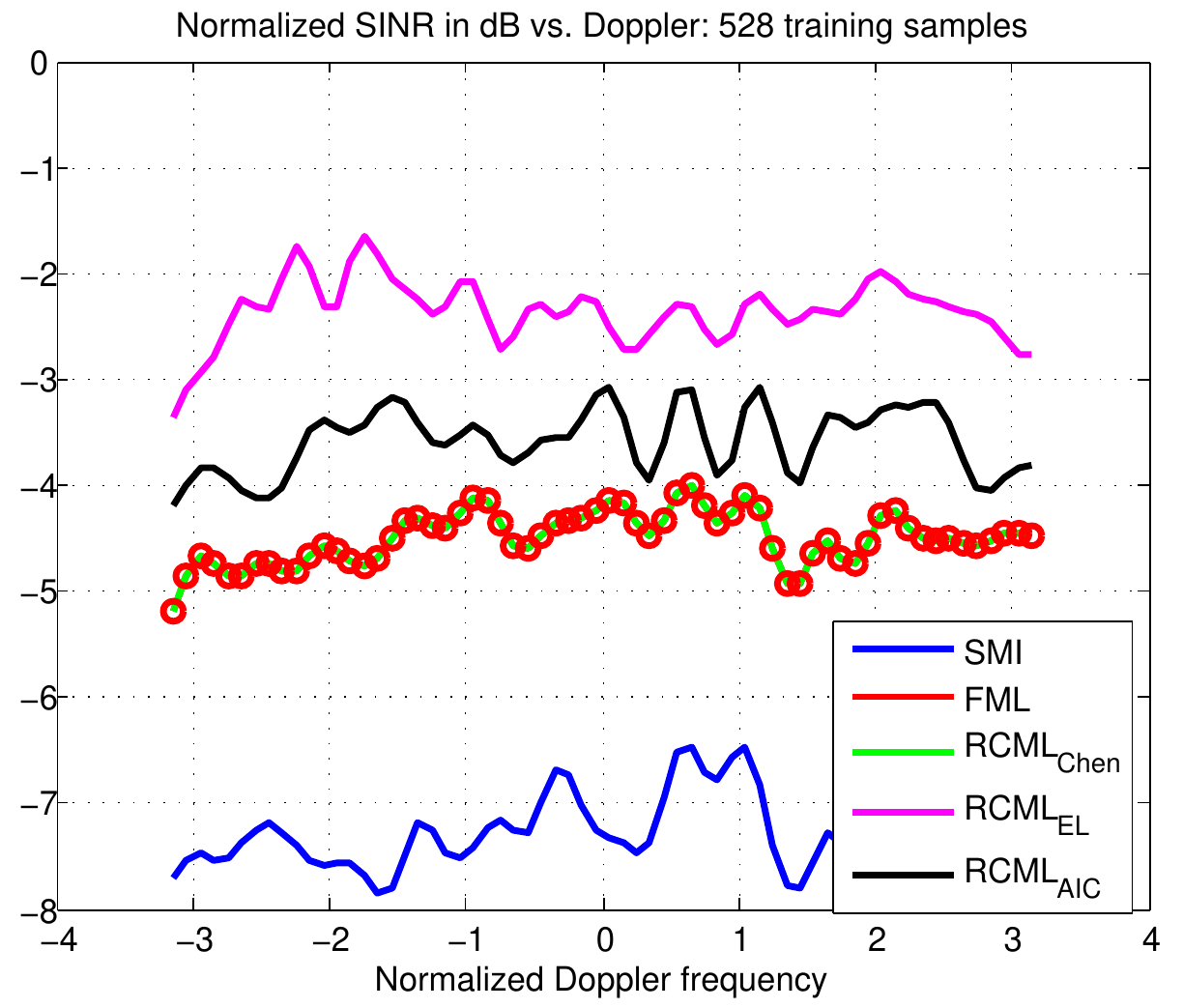}\label{Fig:KASSPER_dop_nonhomo_528}}\\
\subfloat[]{\includegraphics[scale=0.5]{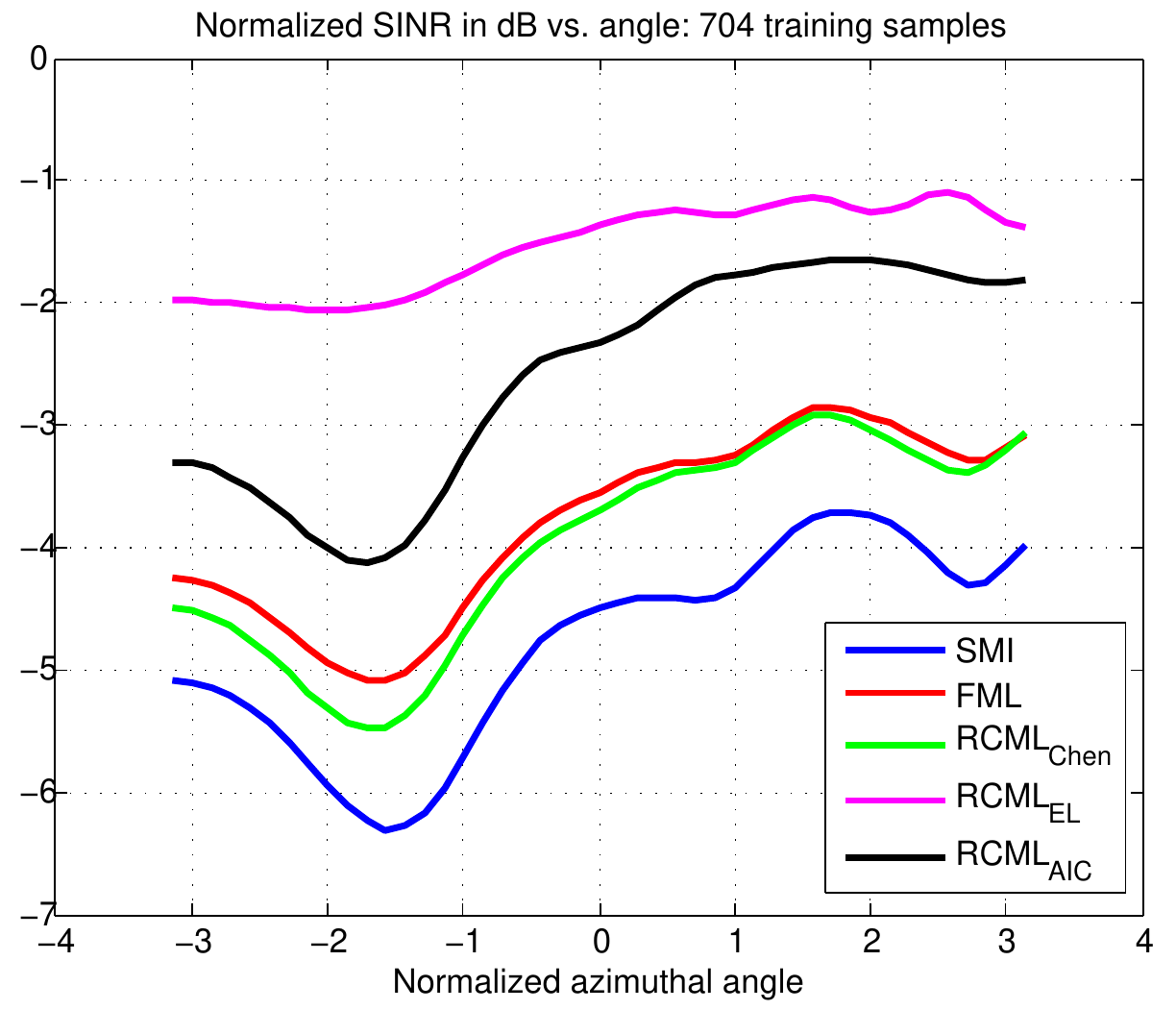}\label{Fig:KASSPER_angle_nonhomo_704}}
\hfil
\subfloat[]{\includegraphics[scale=0.5]{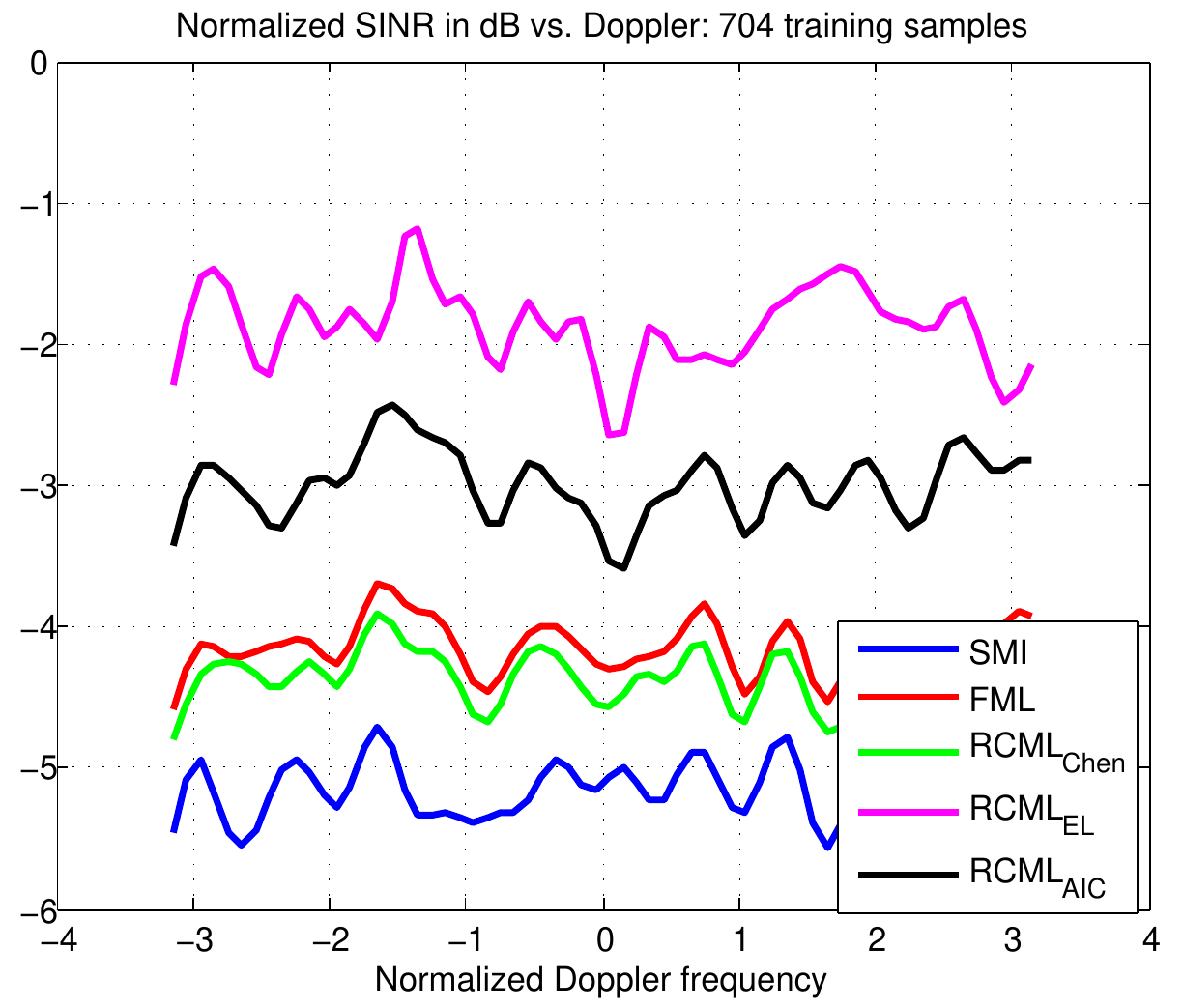}\label{Fig:KASSPER_dop_nonhomo_704}}
\end{center}
\caption{Normalized SINR versus azimuthal angle and Doppler frequency for the KASSPER data set. The case of 50\% of corrupted training data. (a) and (b) for $K=N=352$, (c) and (d) for $K=1.5N=528$, and (e) and (f) for $K=2N=704$}
\label{Fig:KASSPER_rank_nonhomogeneous}
\end{figure*}

First, we compare the rank estimation method proposed in Section \ref{Sec:Rankonly} with alternative algorithms including SMI, FML, AIC, and Chen's algorithm. We plot the normalized SINR (in dB) versus the number of training samples, 20, 30, and 40 in Fig. \ref{Fig:Simulation_rank} for the simulation model. For this experiment, the parameters used are $J=3$, $\beta_i = [0.2, 0, 0.3]$, $\sigma_i  = [10, 100, 1000]$, $\phi_i = [20^\circ, 40^\circ, 60^\circ]$, and $\sigma_a = 1$. The initial rank for Algorithm 1 is the number of jammers ($J=3$). The SINR values are obtained by averaging SINR values from 500 Monte Carlo trials. It is shown that the SINR values increases monotonically as $K$ increases. Fig.\ \ref{Fig:Simulation_rank} reveals that \RCMLEL exhibits the best performance in all training regimes. Particularly, the difference between \RCMLEL and other methods increases when training samples are limited. Table \ref{Tb:RankEstimated} shows the values of the rank estimated by the compared methods. Note that the ranks of SMI and FML are just output of the covariance estimate since they do not estimate the rank. In our simulation model, the true rank is 5 and the rank estimated by \RCMLEL is closer to the true rank.

Fig. \ref{Fig:KASSPER_rank} shows the normalized SINR values for various number of training samples for the KASSPER data set. We plot the averaged SINR values in decibel over either azimuth angle or Doppler frequency domain. The left and right column show the results for angle and Doppler, respectively. We use the rank given by Brennan rule, i.e.\ $M + P - 1 = 42$,  as the initial guess for Algorithm 1. Similar to the results for the simulation model, \RCMLEL outperforms competing methods in all training regimes. Table \ref{Tb:RankEstimated} confirms that the rank predicted via \RCMLEL is closer to the true rank (43 in this case).

\textbf{Realistic case of contaminated observations:} In practice, homogeneous training samples are hard to obtain and a subset of the received signals is often corrupted by outliers resembling a target of interest. Therefore, it is meaningful to compare the performance for nonhomogeneous observation to investigate which algorithm indeed works well and is robust in practice. In this case, the training observations are given by
\be
\left\{\begin{array}{ll}
\mb z = \alpha \mb s + \mb d & \text{when corrupted}\\
\mb z = \mb d & \text{otherwise}
\end{array}\right.
\ee
where $\mb s$ and $\mb d$ represent a target component and the disturbance vector, respectively. Fig. \ref{Fig:KASSPER_rank_nonhomogeneous} shows the normalized SINR values when a half of the training samples contain $\mb s$ with $\alpha = 50$. The gaps between \RCMLEL and the others are bigger than those in Fig. \ref{Fig:KASSPER_rank}. Unsurprisingly, all methods fare worse in the case of corrupted data. However, the drop in \RCMLEL is much smaller than that of competing methods. Notably, in this realistic case of heterogenous or corrupted training, the \RCMLEL now offers a clear advantage over $\text{RCML}_\text{AIC}$. This is further corraborated by the results in Table \ref{Tb:RankEstimated}, which shows that the AIC significantly over-estimates the clutter rank in heterogeneous data than in the homogeneous case leading to the performance degradation.

\subsection{Rank and noise power constraints}
\label{Sec:ResultBoth}

\begin{figure}
\centering
\includegraphics[scale=0.5]{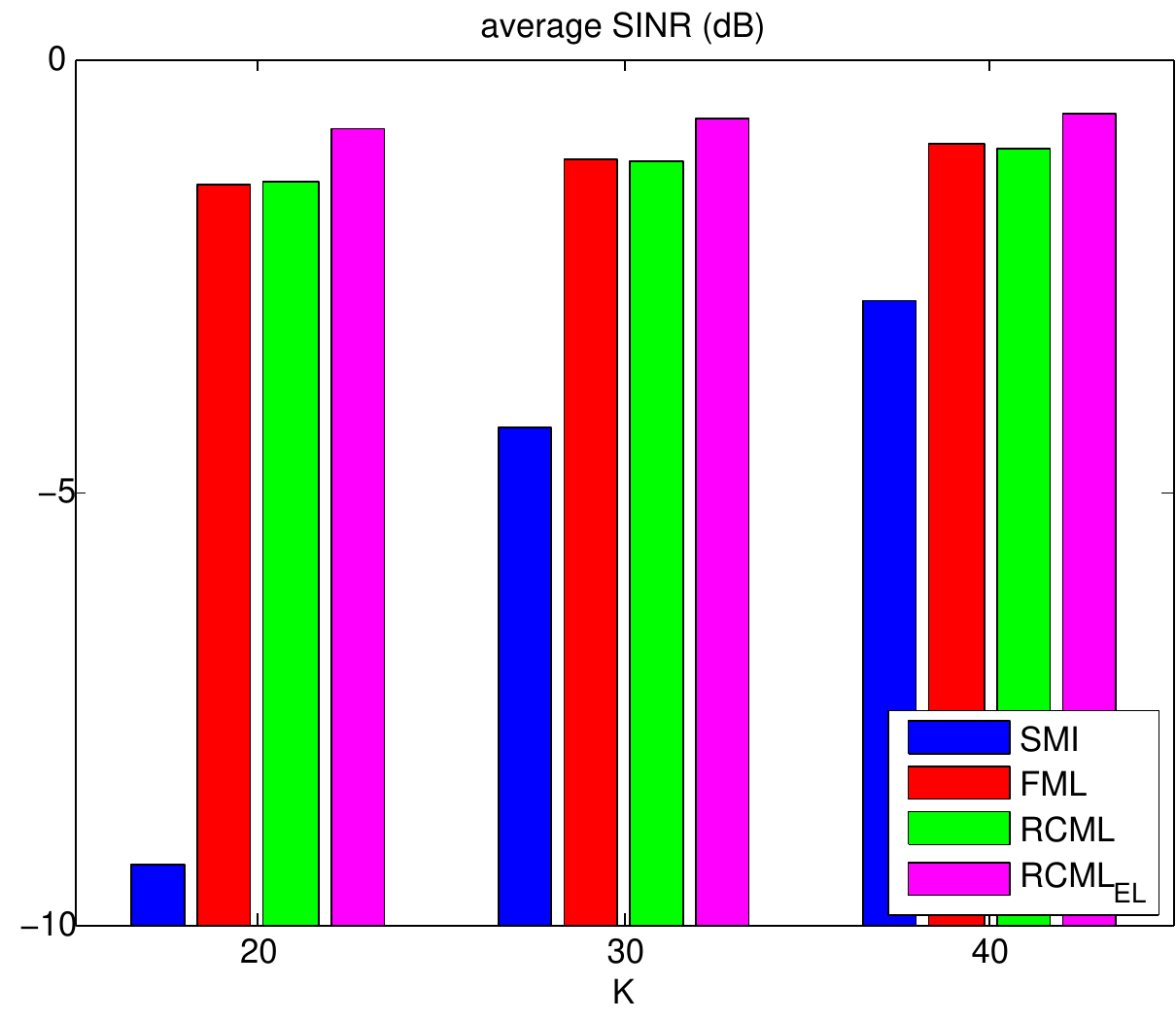}\label{Fig:SINR_Simulation_both}
\caption{Normalized SINR in dB versus number of training samples $K$ $(N=20)$ for the simulation model.}
\label{Fig:Simulation_both}
\end{figure}

\begin{figure*}[!t]
\begin{center}
\subfloat[]{\includegraphics[scale=0.5]{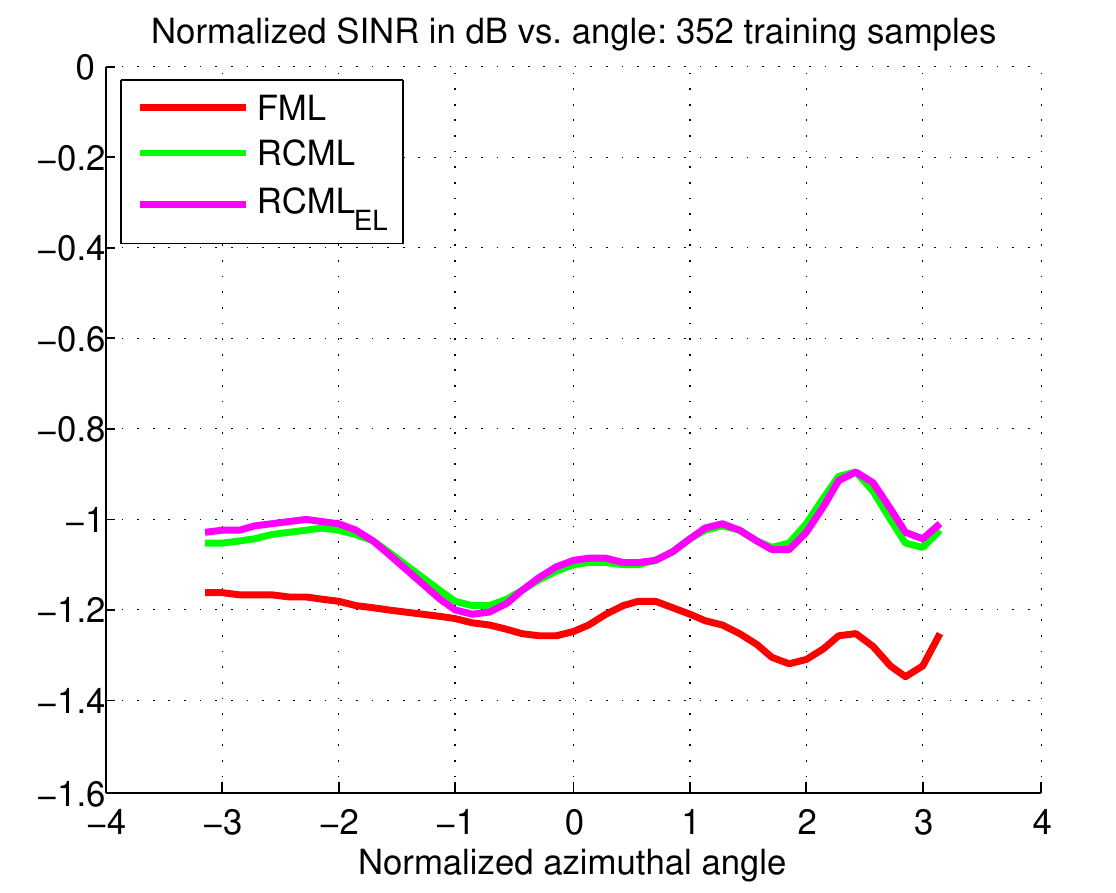}\label{Fig:SINR_angle_352_both}}
\hfil
\subfloat[]{\includegraphics[scale=0.5]{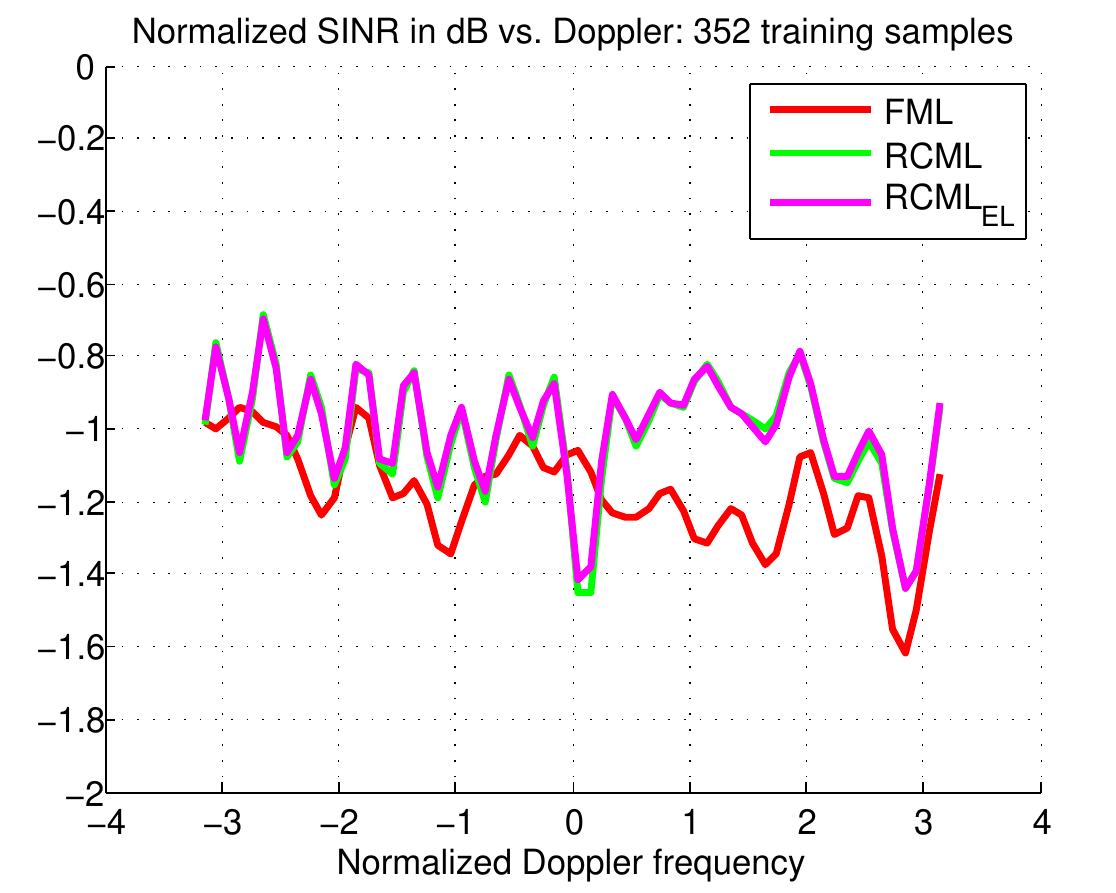}\label{Fig:SINR_dop_352_both}}\\
\subfloat[]{\includegraphics[scale=0.5]{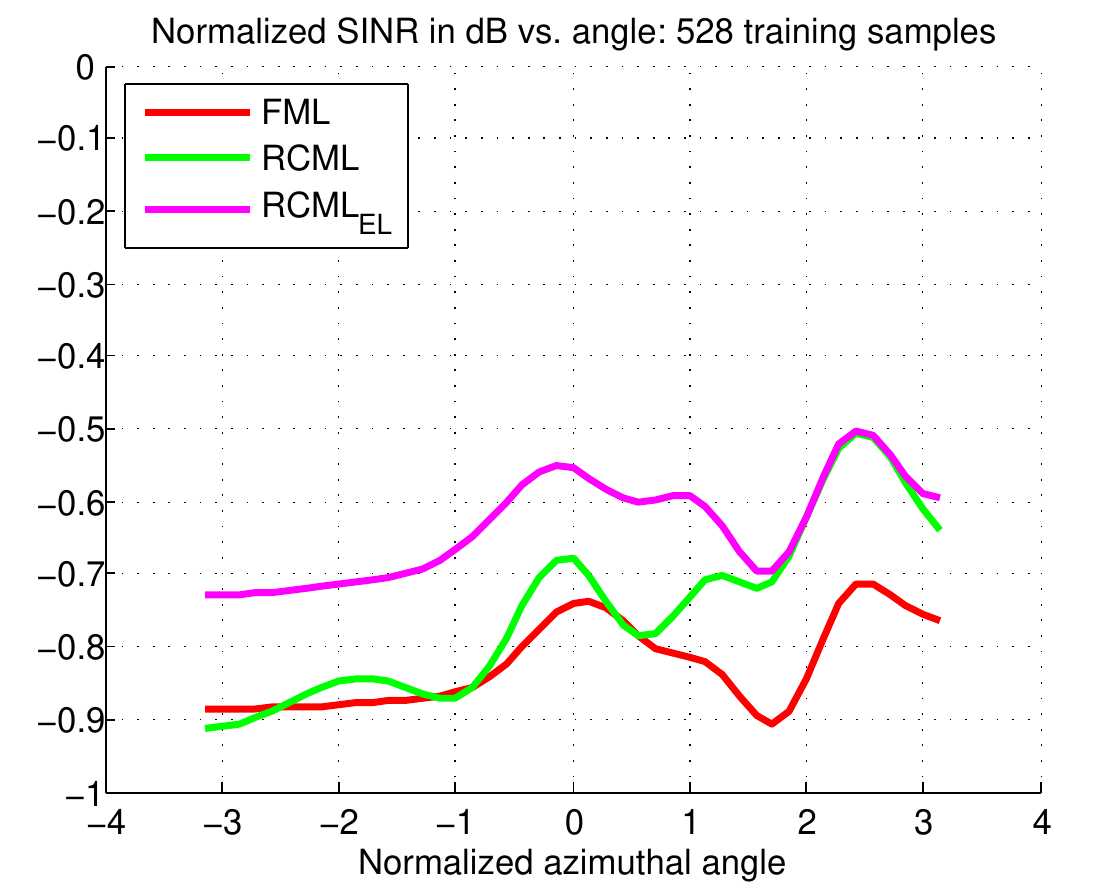}\label{Fig:SINR_angle_528_both}}
\hfil
\subfloat[]{\includegraphics[scale=0.5]{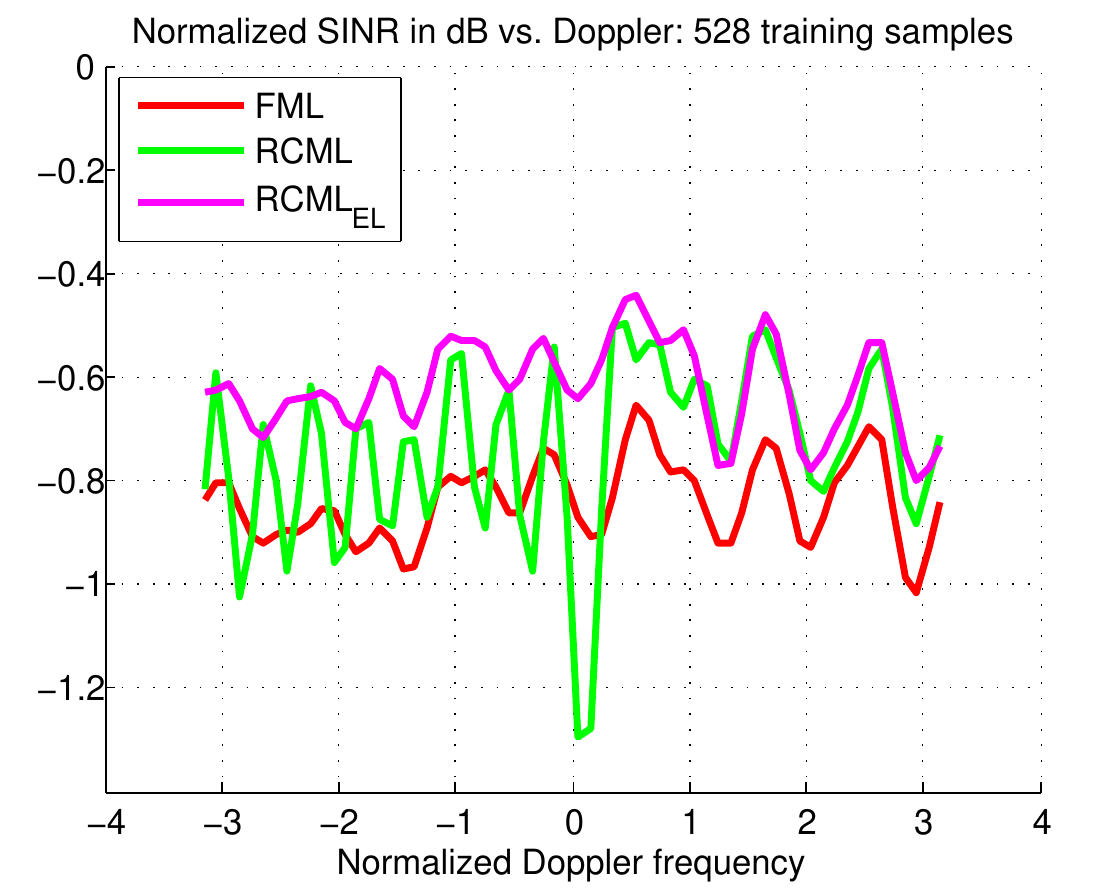}\label{Fig:SINR_dop_528_both}}\\
\subfloat[]{\includegraphics[scale=0.5]{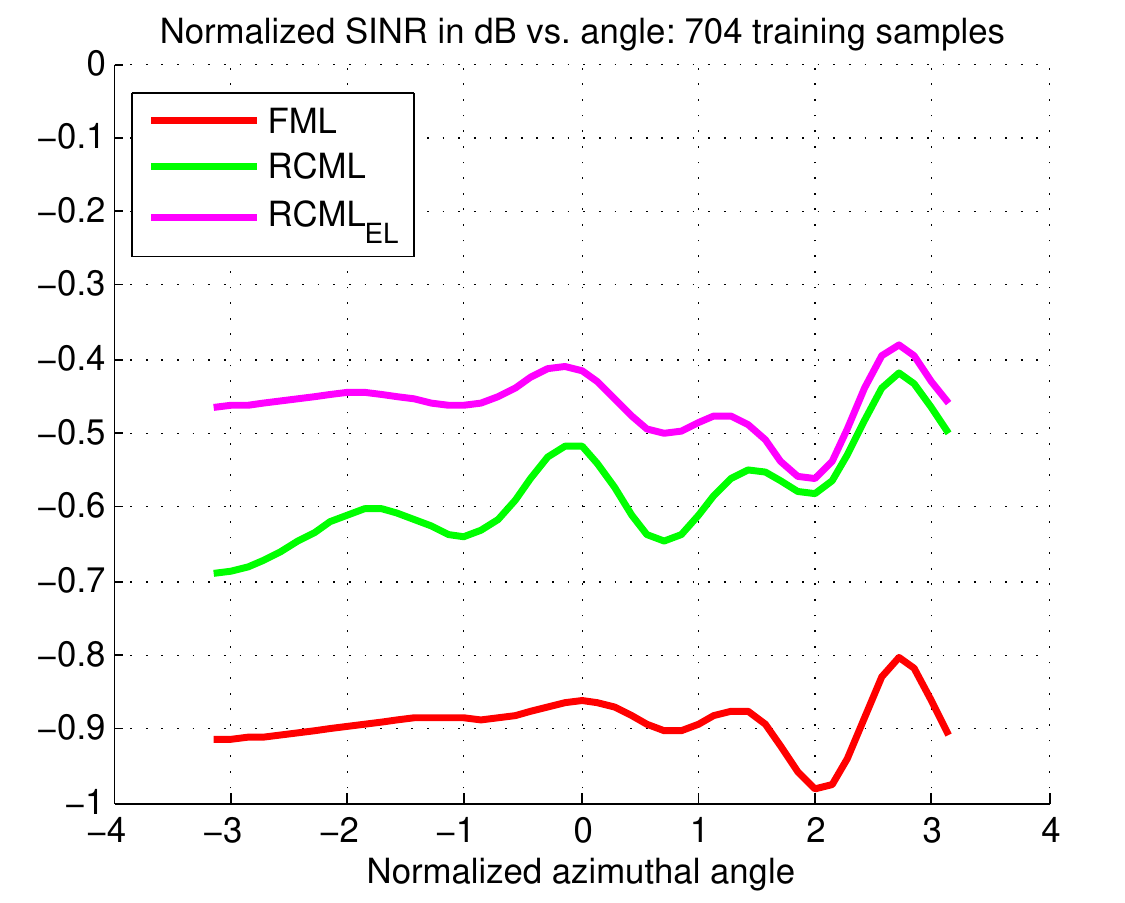}\label{Fig:SINR_angle_704_both}}
\hfil
\subfloat[]{\includegraphics[scale=0.5]{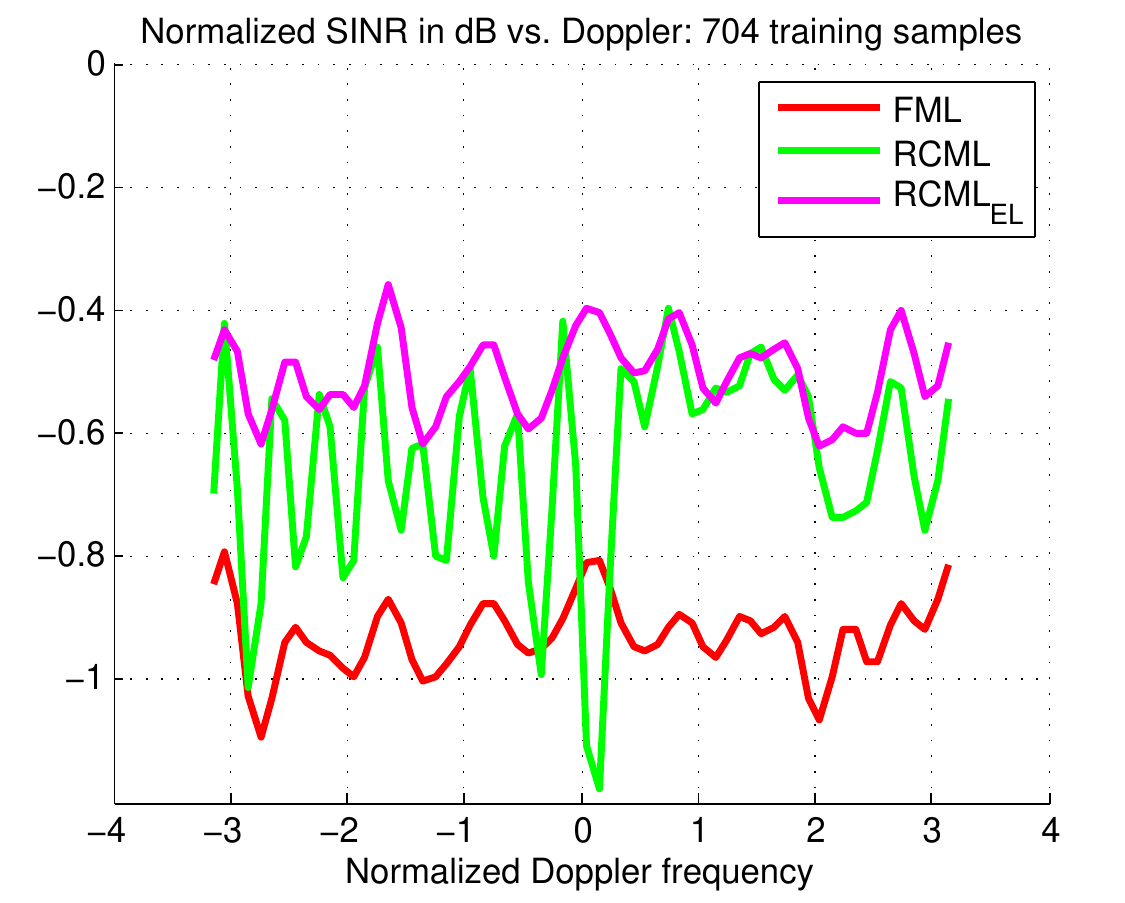}\label{Fig:SINR_dop_704_both}}
\end{center}
\caption{Normalized SINR versus azimuthal angle and Doppler frequency for the KASSPER data set. (a) and (b) for $K=N=352$, (c) and (d) for $K=1.5N=528$, and (e) and (f) for $K=2N=704$.}
\label{Fig:KASSPER_both}
\end{figure*}

In this case, we assume that both the rank and the noise power are unknown for both the simulation model and the KASSPER data set. Since the previous works such as AIC and Chen's algorithm are for only estimating the rank and can not be extended to estimate both the rank and the noise power, we compare the proposed EL method with the sample covariance, FML, and the RCML estimator with a prior knowledge of the rank. For the RCML estimator, we employ the number of jammers $(r=3)$ and the Brennan rule $(r=42)$ as the clutter rank for the simulation model and the KASSPER data set, respectively. In addition, since the FML method requires a prior knowledge of the noise power, we calculate and use the maximum likelihood estimate of the noise power for a rank given by a prior knowledge for the FML.

Fig. \ref{Fig:Simulation_both} shows the performance of various estimators in the sense of the normalized SINR values for the simulation model. Similar to the case of only rank estimation, the \RCMLEL show the best performance in all training regimes.

Fig. \ref{Fig:KASSPER_both} shows the performance of the methods in terms of the normalized output SINR for the KASSPER data set. \RCMLEL is slightly better than the RCML estimator using the rank by Brennan rule. This is expected because for the KASSPER data set Brennan rule predicts a rank very close to the true rank.

\subsection{Condition number constraint}
\label{Sec:ResultConditionNumber}

\begin{table*}[!t]
\begin{center}
\subfloat[]{             \begin{tabular}{|c|c|c|c|c|c|}
               \hline
               $\sigma^2$ & K & SMI & FML & CNCML & $\text{CNCML}_\text{EL}$\\% & $\text{CNCML}_\text{true}$ \\
               \hline
               & 20 & -9.3785 & -0.5195 & -0.5212 & \textbf{-0.4822}\\% & -0.5200 \\
               -5 &30 & -4.2579 & \textbf{-0.4242} & -0.4257 & -0.4256\\% & -0.4250 \\
               & 40 & -2.7424 & \textbf{-0.3460} & -0.3476 & -0.3476\\% & -0.3468 \\
               \hline
                \hline
               & 20 & -9.3196 & -0.5511 & -0.5521 & \textbf{-0.5141}\\% & -0.5515 \\
               0 & 30 & -4.2276 & \textbf{-0.4202} & -0.4221 & -0.4220\\% & -0.4210 \\
               & 40 & -2.7649 & \textbf{-0.3513} & -0.3530 & -0.3528\\% & -0.3521 \\
               \hline
              \hline
               & 20 & -9.0922 & -0.5269 & -0.5279 & \textbf{-0.4875}\\% & -0.5272 \\
               5 & 30 & -4.2172 & \textbf{-0.4348} & -0.4364 & -0.4362\\% & -0.4357 \\
               & 40 & -2.7300 & \textbf{-0.3484} & -0.3503 & -0.3505\\% & -0.3493 \\
               \hline
               \hline
               & 20 & -9.3511 & -0.5355 & -0.5305 & \textbf{-0.4998}\\% & -0.5360 \\
               10 & 30 & -4.1955 & \textbf{-0.4164} & -0.4180 & -0.4175\\% & -0.4175 \\
               & 40 & -2.7491 & \textbf{-0.3501} & -0.3515 & -0.3518\\% & -0.3509 \\
               \hline
             \end{tabular}\label{Tb:CN_simulation_a}}
             \hfill
\subfloat[]{              \begin{tabular}{|c|c|c|c|c|c|}
               \hline
               $\sigma^2$ & K & SMI & FML & CNCML & $\text{CNCML}_\text{EL}$\\% & $\text{CNCML}_\text{true}$ \\
               \hline
               & 20 & -9.3069 & -1.7371 & \textbf{-1.7322} & -1.7358\\% & -1.7350 \\
               -5 & 30 & -4.1795 & -1.2399 & -1.2388 & \textbf{-1.2347}\\% & -1.2397 \\
               & 40 & -2.7535 & -0.9496 & -0.9492 & \textbf{-0.9456}\\% & -0.9493 \\
               \hline
                \hline
               & 20 & -9.1354 & -1.6944 & \textbf{-1.6928} & -1.7027\\% & -1.6940 \\
               0 & 30 & -4.2345 & -1.2986 & -1.2987 & \textbf{-1.2955}\\% & -1.2990 \\
               & 40 & -2.7545 & -1.0041 & -1.0043 & \textbf{-1.0023}\\% & -1.0046 \\
               \hline
              \hline
               & 20 & -9.2524 & -1.3976 & -1.4016 & \textbf{-1.3244}\\% & -1.4000 \\
               5 & 30 & -4.2309 & -1.0737 & -1.0784 & \textbf{-1.0666}\\% & -1.0762 \\
               & 40 & -2.7523 & -0.8848 & -0.8876 & \textbf{-0.8818}\\% & -0.8866 \\
               \hline
               \hline
               & 20 & -9.3660 & -1.2567 & -1.2569 & \textbf{-1.2115}\\% & -1.2570 \\
               10 & 30 & -4.3013 & -0.9526 & -0.9545 & \textbf{-0.9450}\\% & -0.9537 \\
               & 40 & -2.7350 & -0.7171 & -0.7197 & \textbf{-0.7139}\\% & -0.7186 \\
               \hline
             \end{tabular}  \label{Tb:CN_simulation_b} }\\
             \subfloat[]{
\begin{tabular}{|c|c|c|c|c|c|c|}
               \hline
               $\sigma^2$ & K & SMI & FML & CNCML & $\text{CNCML}_\text{EL}$\\% & $\text{CNCML}_\text{true}$ \\
               \hline
               & 20 & -9.3702 & -0.5340 & -0.5349 & \textbf{-0.4925}\\% & -0.5340 \\
               -5 & 30 & -4.2791 & \textbf{-0.4302} & -0.4316 & -0.4315\\% & -0.4308 \\
               & 40 & -2.7856 & \textbf{-0.3493} & -0.3510 & -0.3509\\% & -0.3501 \\
               \hline
                \hline
               & 20 & -9.2898 & -0.5485 & -0.5501 & \textbf{-0.5104}\\% & -0.5491 \\
               0 & 30 & -4.2648 & \textbf{-0.4202} & -0.4219 & -0.4220\\% & -0.4209 \\
               & 40 & -2.7274 & \textbf{-0.3604} & -0.3621 & -0.3621\\% & -0.3611 \\
               \hline
              \hline
               & 20 & -9.0582 & -0.5318 & -0.5328 & \textbf{-0.4899}\\% & -0.5322 \\
               5 & 30 & -4.1548 & \textbf{-0.4142} & -0.4155 & -0.4152\\% & -0.4149 \\
               & 40 & -2.7655 & \textbf{-0.3515} & -0.3531 & -0.3533\\% & -0.3521 \\
               \hline
               \hline
               & 20 & -9.3632 & -0.5352 & -0.5363 & \textbf{-0.4974}\\% & -0.5360 \\
               10 & 30 & -4.2728 & \textbf{-0.4328} & -0.4348 & -0.4349\\% & -0.4337 \\
               & 40 & -2.7577 & \textbf{-0.3538} & -0.3554 & -0.3547\\% & -0.3547 \\
               \hline
             \end{tabular}\label{Tb:CN_simulation_c}}\hfill
\subfloat[]{\begin{tabular}{|c|c|c|c|c|c|}
               \hline
               $\sigma^2$ & K & SMI & FML & CNCML & $\text{CNCML}_\text{EL}$\\% & $\text{CNCML}_\text{true}$ \\
               \hline
               & 20 & -9.0316 & -1.7161 & \textbf{-1.7131} & -1.7634\\% & -1.7150 \\
               -5 & 30 & -4.1465 & -1.1704 & -1.1691 & \textbf{-1.1659}\\% & -1.1693 \\
               & 40 & -2.7727 & -0.9390 & -0.9384 & \textbf{-0.9351}\\% & -0.9387 \\
               \hline
                \hline
               & 20 & -9.2091 & -1.6706 & -1.6701 & \textbf{-1.6674}\\% & -1.6706 \\
               0 & 30 & -4.2004 & -1.2681 & -1.2682 & \textbf{-1.2633}\\% & -1.2682 \\
               & 40 & -2.7423 & -1.0102 & -1.0117 & \textbf{-1.1009}\\% & -1.0106 \\
               \hline
              \hline
               & 20 & -9.3538 & -1.3980 & -1.4028 & \textbf{-1.3216}\\% & -1.4004 \\
               5 & 30 & -4.2203 & -1.0869 & -1.0910 & \textbf{-1.0785}\\% & -1.0889 \\
               & 40 & -2.7079 & -0.8694 & -0.8721 & \textbf{-0.8666}\\% & -0.8713 \\
               \hline
               \hline
               & 20 & -9.221 & -1.2446 & -1.2455 & \textbf{-1.1982}\\% & -1.2452 \\
               10 & 30 & -4.2116 & -0.9428 & -0.9460 & \textbf{-0.9382}\\% & -0.9444 \\
               & 40 & -2.7563 & -0.7235 & -0.7264 & \textbf{-0.7226}\\% & -0.7253 \\
               \hline
             \end{tabular}\label{Tb:CN_simulation_d}}\\
             \subfloat[]{            \begin{tabular}{|c|c|c|c|c|c|}
               \hline
               $\sigma^2$ & K & SMI & FML & CNCML & $\text{CNCML}_\text{EL}$\\% & $\text{CNCML}_\text{true}$ \\
               \hline
               & 20 & -9.2679 & -1.1593 & -1.1616 & \textbf{-1.1150}\\% & -1.1610 \\
               -5 & 30 & -4.2234 & -0.9262 & -0.9286 & \textbf{-0.9242}\\% & -0.9278 \\
               & 40 & -2.8271 & \textbf{-0.7705} & -0.7729 & -0.7712\\% & -0.7723 \\
               \hline
                \hline
               & 20 & -9.2934 & -0.9052 & -0.9051 & \textbf{-0.8422}\\% & -0.9059 \\
               0 & 30 & -4.1617 & -0.6909 & -0.6920 & \textbf{-0.6862}\\% & -0.6924 \\
               & 40 & -2.7387 & -0.5711 & -0.5724 & \textbf{-0.5676}\\% & -0.5725 \\
               \hline
              \hline
               & 20 & -9.4154 & -0.8398 & -0.8334 & \textbf{-0.7909}\\% & -0.8399 \\
               5 & 30 & -4.2284 & -0.6273 & -0.6231 & \textbf{-0.6070}\\% & -0.6278 \\
               & 40 & -2.7208 & -0.5034 & -0.5022 & \textbf{-0.4945}\\% & -0.5046 \\
               \hline
               \hline
               & 20 & -9.1447 & -0.7388 & -0.7225 & \textbf{-0.6815}\\% & -0.7392 \\
               10 & 30 & -4.2046 & -0.5931 & -0.5803 & \textbf{-0.5535}\\% & -0.5931 \\
               & 40 & -2.7241 & -0.4821 & -0.4738 & \textbf{-0.4576}\\% & -0.4827 \\
               \hline
             \end{tabular}\label{Tb:CN_simulation_e}}\hfill
             \renewcommand{\arraystretch}{1.3}
             \subfloat[]{\begin{tabular}{|c|c|c|c|c|}
               \hline
               & $J$ & $\sigma_J^2$ & $\phi$ & $B_f$\\
               \hline
               (a) & 1 & 30 & 20$^{\circ}$ & 0\\
                \hline
               (b) & 1 & 30 & 20$^{\circ}$ & 0.3\\
              \hline
               (c) & 3 & [30 30 30] & [20$^{\circ}$ 40$^{\circ}$ 60$^{\circ}$] & [0 0 0]\\
               \hline
               (d) & 3 & [30 30 30] & [20$^{\circ}$ 40$^{\circ}$ 60$^{\circ}$] & [0.3 0.3 0.3]\\
               \hline
               (e) & 3 & [10 20 30] & [20$^{\circ}$ 40$^{\circ}$ 60$^{\circ}$] & [0.2 0 0.3]\\
               \hline
             \end{tabular}\label{Tb:CN_parameter}}
        \end{center}
        \caption{Normalized SINR for various values of parameters for the simulation model.}
        \label{Tb:CN_simulation}
\end{table*}

\begin{figure*}[!t]
\begin{center}
\subfloat[]{\includegraphics[scale=0.5]{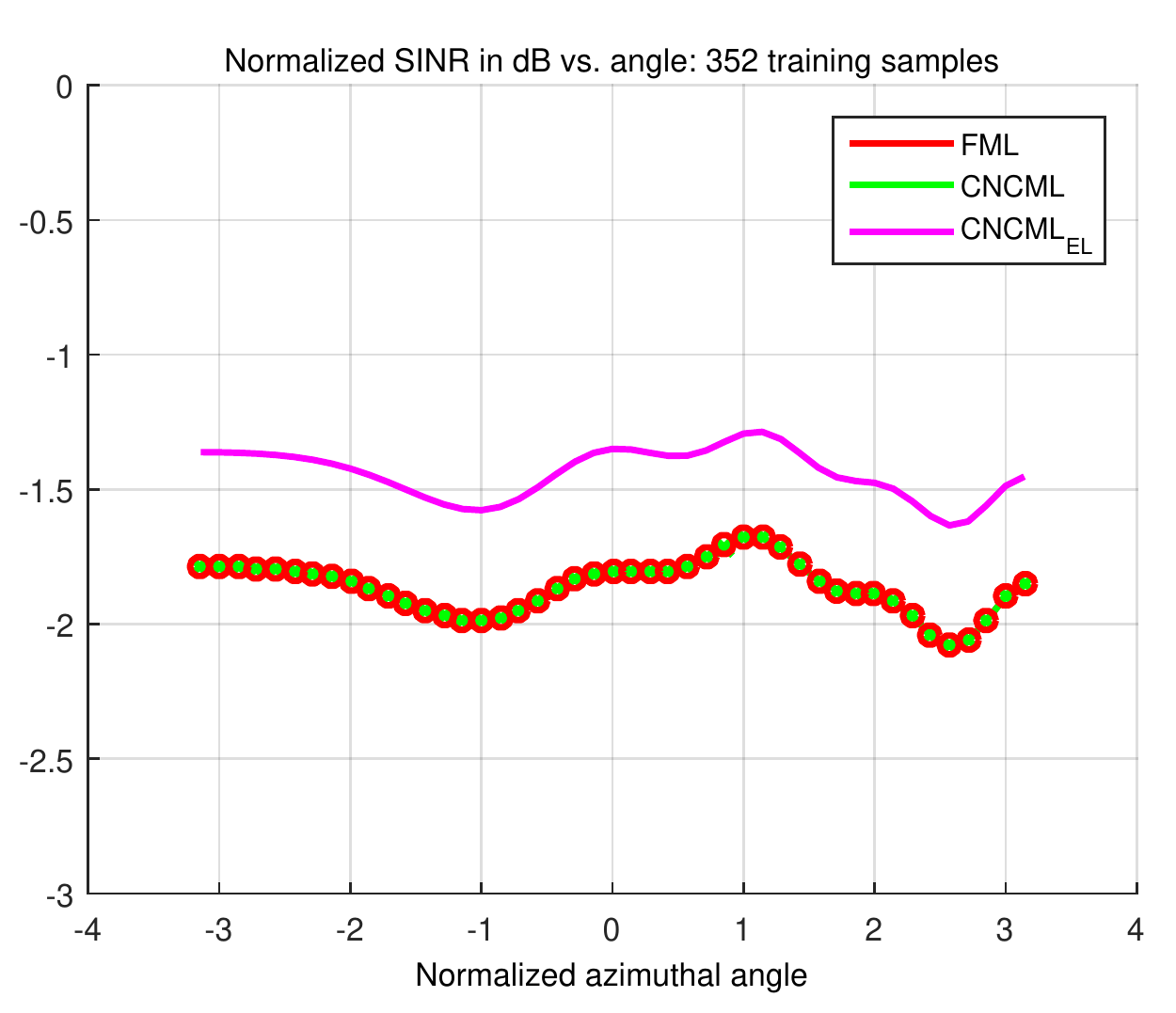}\label{Fig:SINR_angle_352_CN}}
\hfil
\subfloat[]{\includegraphics[scale=0.5]{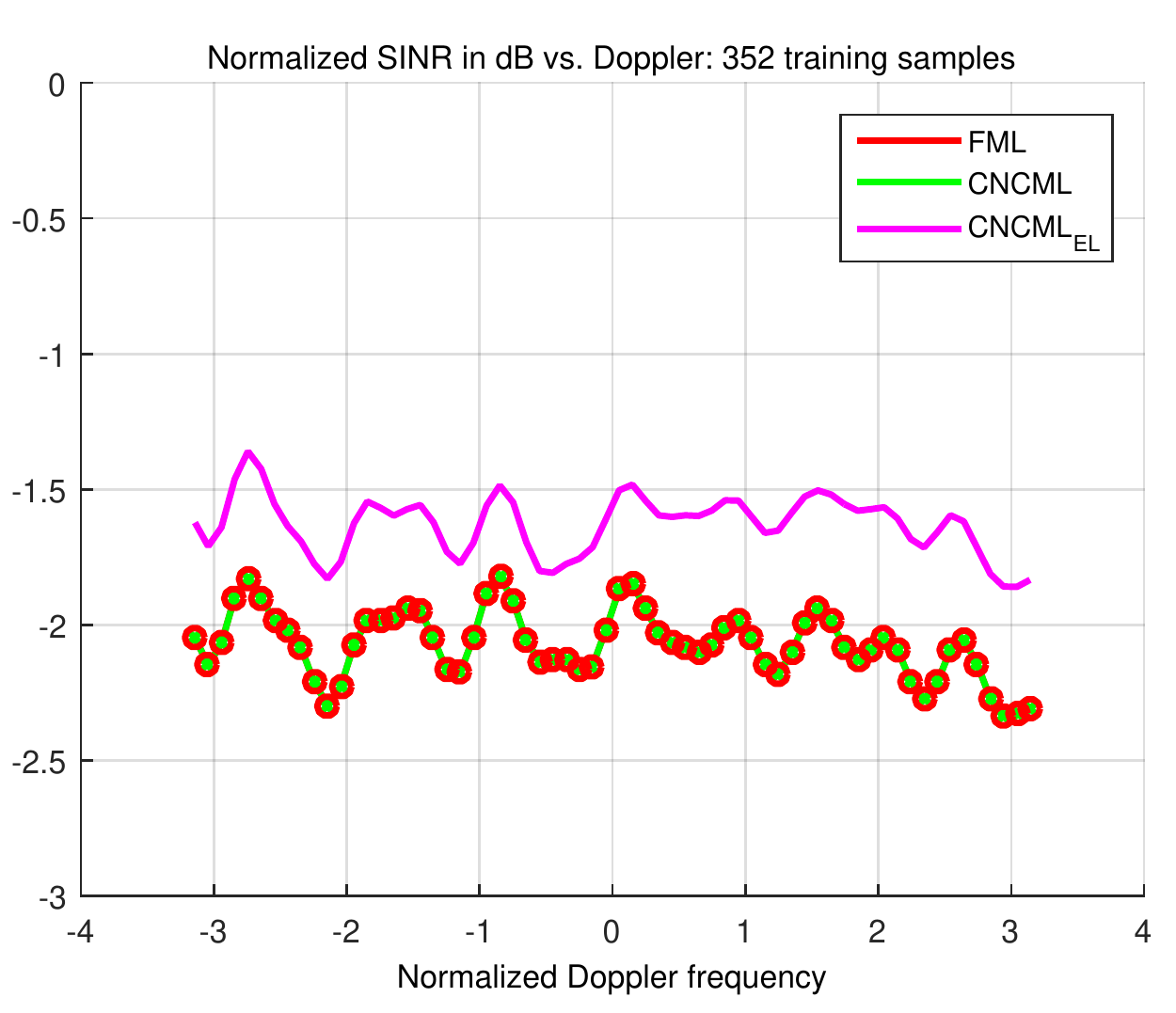}\label{Fig:SINR_dop_352_CN}}\\
\subfloat[]{\includegraphics[scale=0.5]{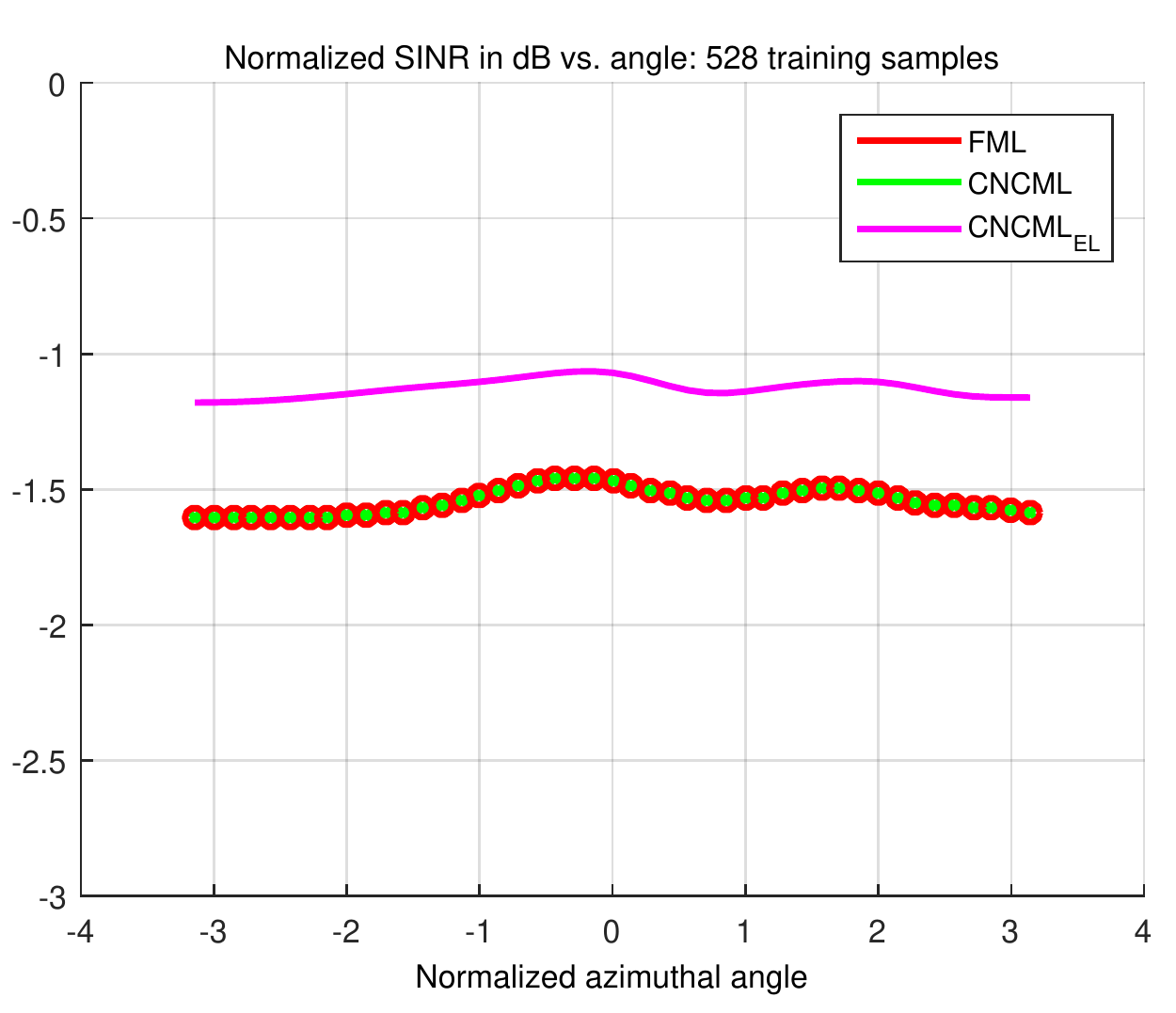}\label{Fig:SINR_angle_528_CN}}
\hfil
\subfloat[]{\includegraphics[scale=0.5]{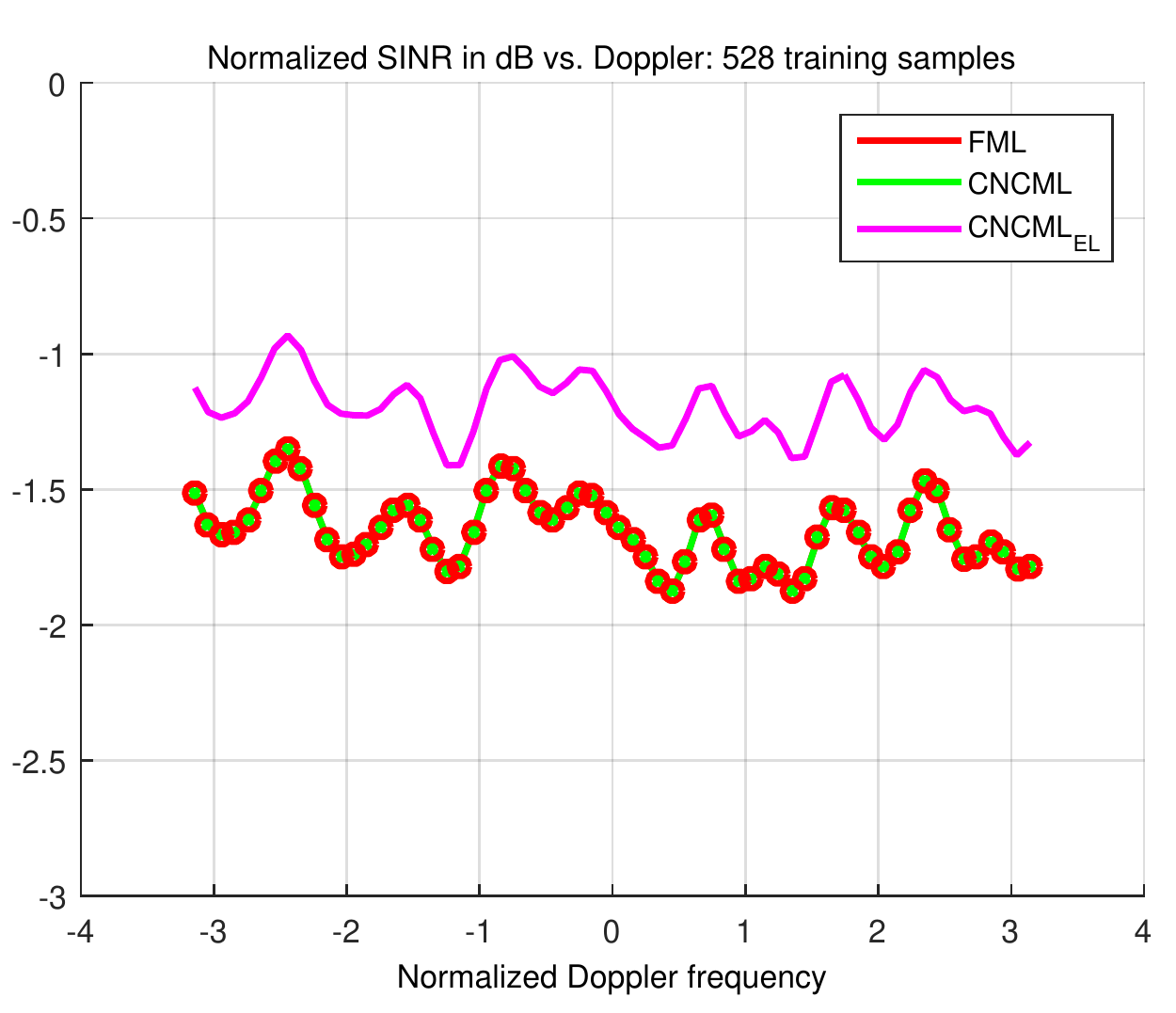}\label{Fig:SINR_dop_528_CN}}\\
\subfloat[]{\includegraphics[scale=0.5]{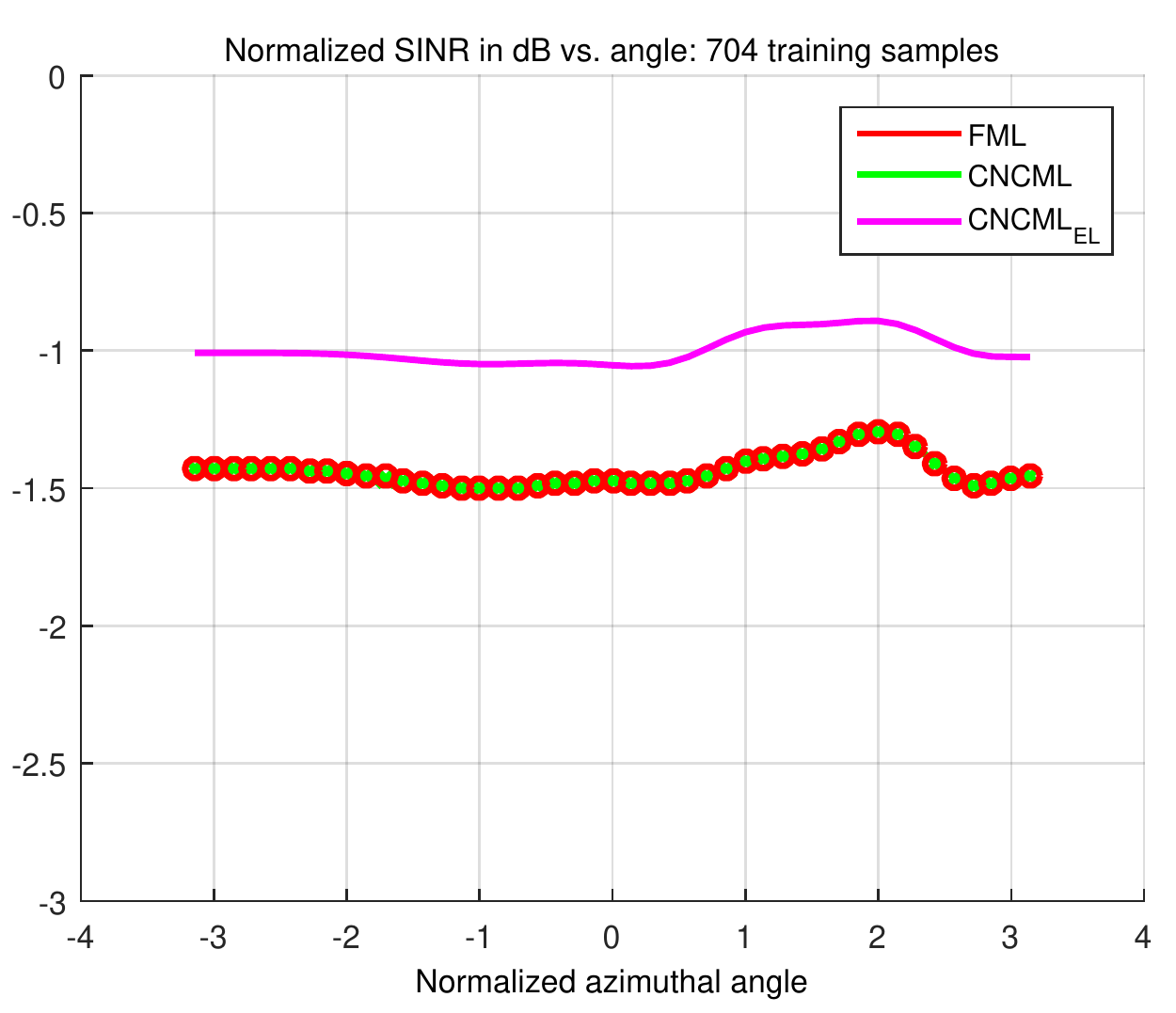}\label{Fig:SINR_angle_704_CN}}
\hfil
\subfloat[]{\includegraphics[scale=0.5]{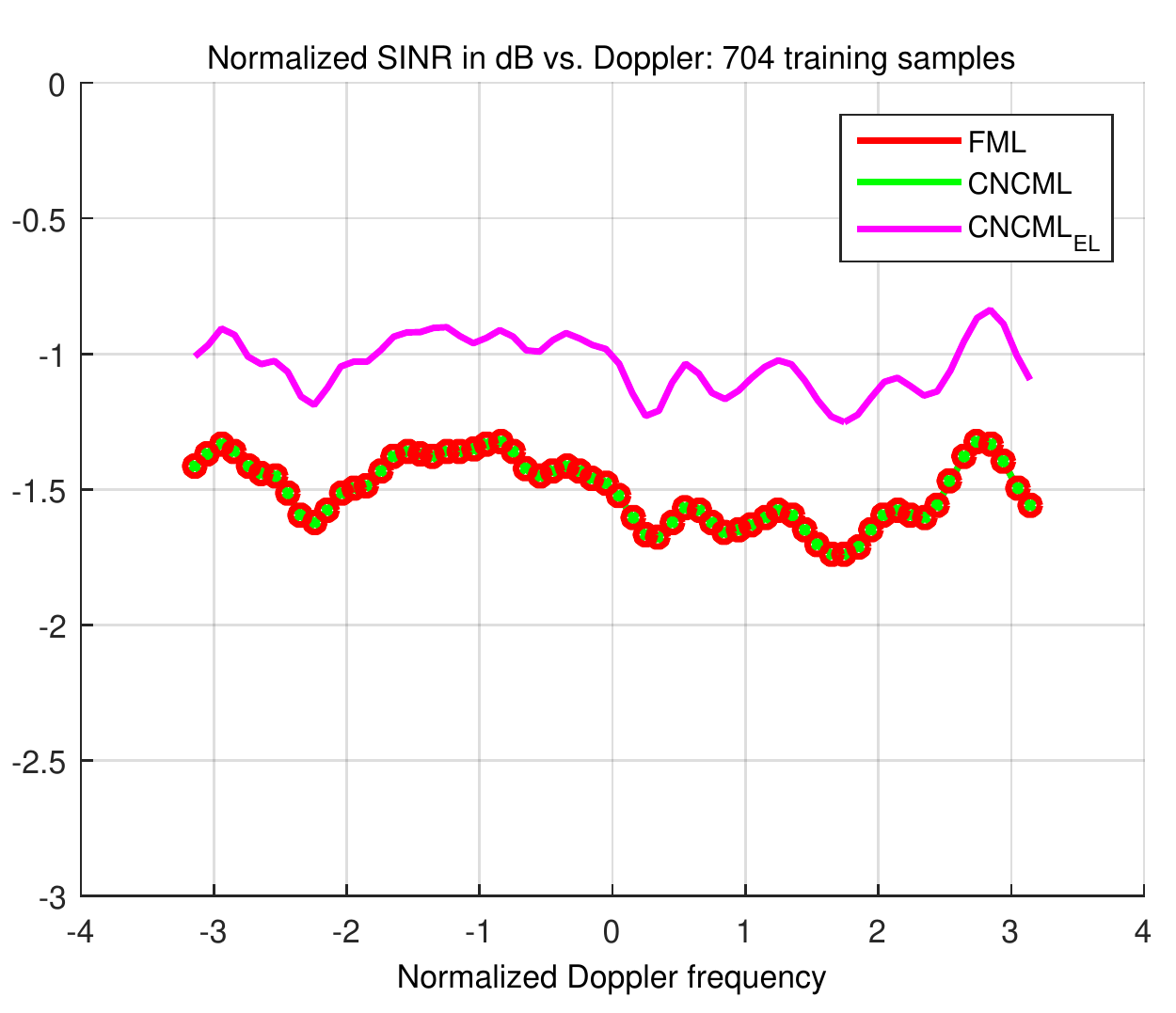}\label{Fig:SINR_dop_704_CN}}
\end{center}
\caption{Normalized SINR versus azimuthal angle and Doppler frequency for the KASSPER data set. (a) and (b) for $K=N=352$, (c) and (d) for $K=1.5N=528$, and (e) and (f) for $K=2N=704$}
\label{Fig:KASSPER_CN}
\end{figure*}

Now we show experimental results for the condition number estimation method proposed in Section \ref{Sec:ConditionNumber}. We compare the proposed method, denoted by \CNCEL\hspace{-1mm}, with three different covariance estimation methods, the sample covariance matrix (SMI), FML, and \CNCML proposed by Aubry \emph{et al.} \cite{Aubry12}.

Table \ref{Tb:CN_simulation} shows the normalized SINR values for the simulation model. We analyze five different scenarios with different parameters of the simulated covariance model given by Eq. \eqref{Eq:SimulationModel}. We use the same parameters as those used in \cite{Aubry12} to evaluate the performances and they are shown in Table \ref{Tb:CN_parameter}.

For the narrowband scenarios ($B_f = 0$) in Table \ref{Tb:CN_simulation_a} and Table \ref{Tb:CN_simulation_c}, \CNCEL outperforms the alternatives for the limited training regime and FML is the best in other training regimes. Note that the gap between \CNCEL and FML (at most 0.002) is much smaller than that of the limited training regime (at least 0.3). On the other hand, for the wideband scenarios in Table \ref{Tb:CN_simulation_b}, Table \ref{Tb:CN_simulation_d}, and Table \ref{Tb:CN_simulation_e}, \CNCEL shows the best performance in most cases.

The experimental results for the KASSPER data set are shown in Fig. \ref{Fig:KASSPER_CN}. We do not plot the sample covariance matrix to clarify the difference among the estimators. In every case, FML and \CNCML  are very close to each other and \CNCEL  is the best estimator.

%Note that \CNCEL  is based on the same algorithm as \CNCML  and differs from \CNCML  in the point that \CNCEL  uses a different condition number which is estimated by the expected likelihood approach. Again, this shows that the expected likelihood criterion is really useful and powerful to estimate parameters which are imperfectly known and leads to an adaptive and robust covariance estimator.

\section{Conclusion}
\label{Sec:Conclusion}

We propose robust covariance estimation algorithms which automatically determine the optimal values of practical constraints via the expected likelihood criterion for radar STAP. Three different cases of practical constraints which is exploited in recent works including the rank constrained ML estimation and the condition number constrained ML estimation are investigated. New analytical results are derived for each case. Uniqueness of the optimal values of the rank constraint and the condition number constraint is formally proved and a closed form solution of the noise level is obtained for a fixed rank. Experimental results show that the estimators with the constraints obtained by the expected likelihood approach outperform state of the art alternatives including those based on maximum likelihood solution of the constraints.

\section*{Appendix}
\label{Sec:Appendix}

\subsection{Proof of Lemma \ref{Lemma1}}
\label{Sec:AppendixA}

First, let $r$ be the largest $i$ such that $d_{i+1} \geq \sigma^2$. Then, from the closed form solution of the RCML estimator, the eigenvalues of the RCML estimator with rank $i$ and $i+1$ for given $i  < r$ will be
\begin{itemize}
  \item $\hat{\mb R}_\text{RCML}(i)$ : $d_1, \; \; d_2,  \; \; \ldots,  \; \; d_i,  \; \; \sigma^2,  \; \; \ldots,  \; \; \sigma^2$
  \item $\hat{\mb R}_\text{RCML}(i+1)$ : $d_1, \; \; d_2,  \; \; \ldots,  \; \; d_i,  \; \; d_{i+1}, \; \; \sigma^2,  \; \; \ldots,  \; \; \sigma^2$
\end{itemize}
since $d_{i+1} \geq \sigma^2$.
Then $\dfrac{d_i}{\lambda_i}$ should be
\begin{itemize}
  \item $\hat{\mb R}_\text{RCML}(i)$ : $1, \; \; 1,  \; \; \ldots,  \; \; 1_i,  \; \; \dfrac{d_{i+1}}{\sigma^2},  \; \; \ldots,  \; \; \dfrac{d_N}{\sigma^2}$
  \item $\hat{\mb R}_\text{RCML}(i+1)$ : $1, \; \; 1,  \; \; \ldots,  \; \; 1_i,  \; \; 1_{i+1}, \; \; \dfrac{d_{i+2}}{\sigma^2},  \; \; \ldots,  \; \; \dfrac{d_N}{\sigma^2}$
\end{itemize}
From Eq. \eqref{Eq:SimplifiedLR}, the LR values of the RCML estimators with the ranks $i$ and $i+1$ are
\be
\label{Eq:LR_r}
\lr(i) = \frac{\dfrac{\exp N}{\sigma^{2(N-i)}}\ds\prod_{k=i+1}^N d_k }{\exp (i + \dfrac{1}{\sigma^2}\ds\sum_{k=i+1}^N d_k)}
\ee
\be
\label{Eq:LR_r+1}
\lr(i+1) = \frac{\dfrac{\exp N}{\sigma^{2(N-i-1)}}\ds\prod_{k=i+2}^N d_k }{\exp (i + 1 +\dfrac{1}{\sigma^2}\ds\sum_{k=i+2}^N d_k)}
\ee

From Eq. \eqref{Eq:LR_r} and Eq. \eqref{Eq:LR_r+1}, we obtain
\bea
\lr(i+1) & = & \frac{\dfrac{\exp N}{\sigma^{2(N-i-1)}}\ds\prod_{k=i+2}^N d_k }{\exp (i + 1 +\dfrac{1}{\sigma^2}\ds\sum_{k=i+2}^N d_k)}\\
& = & \frac{\dfrac{\exp N}{\sigma^{2(N-i)}}\ds\prod_{k=i+1}^N d_k \cdot \frac{\sigma^2}{d_{i+1}}}{\exp (i +\dfrac{1}{\sigma^2}\ds\sum_{k=i+1}^N d_k) \exp(1-\dfrac{d_{i+1}}{\sigma^2})}\\
& = & \lr(i) \cdot \frac{\sigma^2}{d_{i+1}} \cdot \exp(\frac{d_{i+1}}{\sigma^2}-1)\label{Eq:LRrelation}
\eea
Eq. \eqref{Eq:LRrelation} tells us $\lr (i+1)$ can be calculated by multiplying $\lr (i)$ by the coefficient $\dfrac{\sigma^2}{d_{i+1}} \cdot \exp(\dfrac{d_{i+1}}{\sigma^2}-1)$. Fig. \ref{Fig:increase} shows that
\be
\dfrac{\sigma^2}{d_{i+1}} \cdot \exp(\dfrac{d_{i+1}}{\sigma^2}-1) \geq 1
\ee
for all values of $\dfrac{\sigma^2}{d_{i+1}}$. Therefore, it is obvious that
\be
\lr(i+1) \geq \lr(i),
\ee
which means the LR value monotonically increases with respect to $i$.

Now, let's consider the other case, $i \geq r$. In this case, since $d_{i+1} < \sigma^2$, it is easily shown that
\be
\mb R_\text{RCML}(i) = \mb R_\text{RCML}(i+1)
\ee
Therefore,
\be
\lr(i+1) = \lr(i)
\ee

This proves that $\lr(i)$ monotonically increases for all $1 \leq i \leq N$.

\begin{figure}
\centering
\includegraphics[scale=0.5]{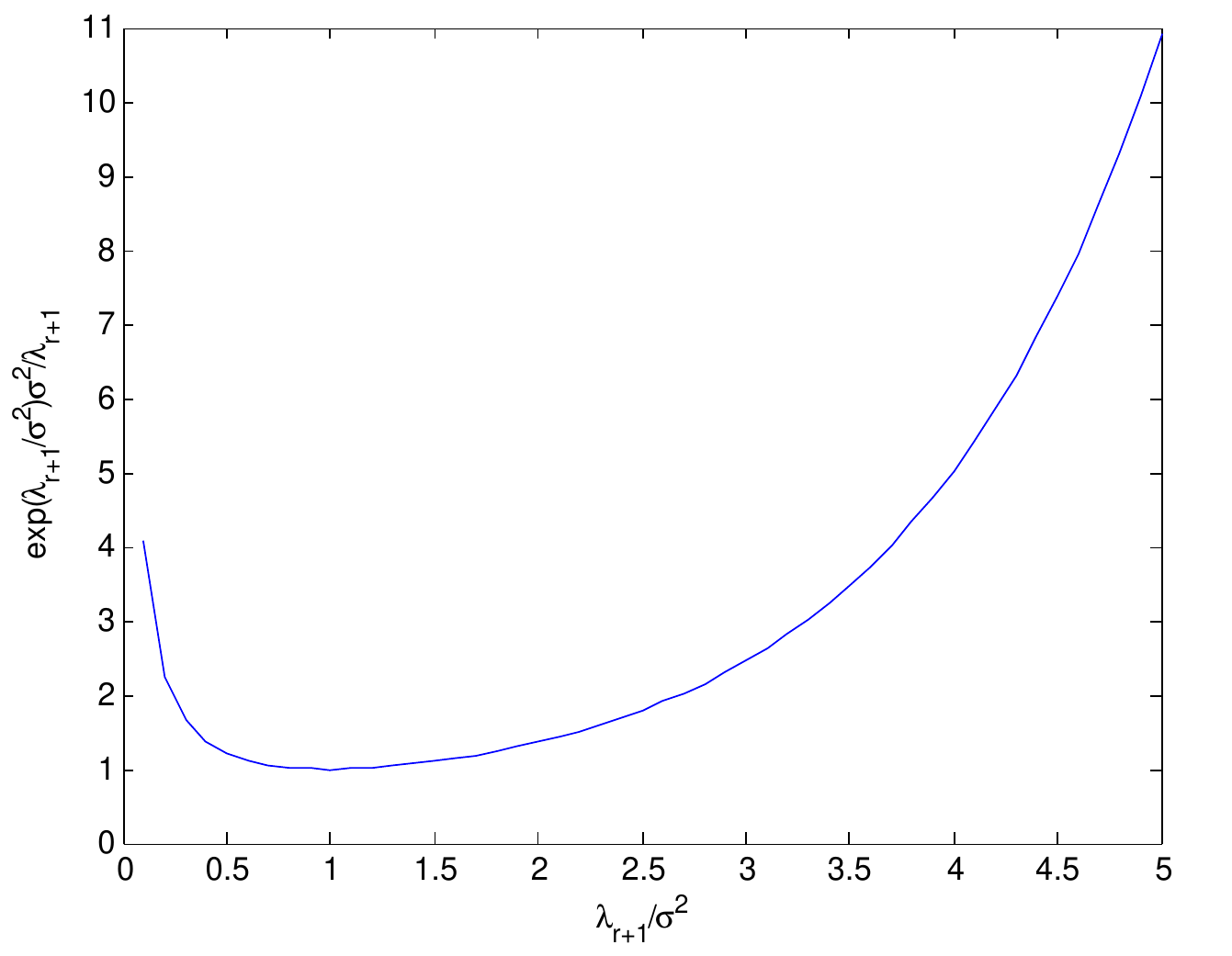}
\caption{The value of the coefficient $\dfrac{\sigma^2}{\lambda_{r+1}} \cdot \exp(\dfrac{\lambda_{r+1}}{\sigma^2}-1)$}
\label{Fig:increase}
\end{figure}

\subsection{Proof of Lemma \ref{Lemma2}}
\label{Sec:AppendixB}

In this section, I investigate the LR values for varying noise level $\sigma^2$ and a given rank $r$. From Eq. \eqref{Eq:LR_r} we obtain the LR value when the rank is $r$,
\be
\label{Eq:LR_sigma}
\lr(\sigma^2) = \frac{\dfrac{\exp N}{\sigma^{2(N-r)}}\ds\prod_{k=r+1}^N d_k }{\exp (r + \dfrac{1}{\sigma^2}\ds\sum_{k=r+1}^N d_k)}
\ee
For simplicity, let $\sigma^2 = t$ then Eq. \eqref{Eq:LR_sigma} can be simplified as
\be
\lr(t) = \dfrac{e^{N-r}\ds\prod_{k=r+1}^N d_k }{t^{N-r} e^{\dfrac{\sum_{k=r+1}^N d_k}{t}}}
\ee
Now let $\ds\sum_{k=r+1}^N d_k = d_s$ and $\ds\prod_{k=r+1}^N d_k = d_p$, then
\bea
\lr(t) & = & \dfrac{e^{N-r}d_p }{t^{N-r} e^{\frac{d_s}{t}}}\\
& = & d_p e^{N-r} t^{r-N} e^{-\frac{d_s}{t}}\label{Eq:LR_t}
\eea
To analyze increasing or decreasing property Eq. \eqref{Eq:LR_t}, I calculate its first derivative. Since $d_p e^{N-r}$ is a positive constant, it does not affect increasing or decreasing of the function. Therefore,
\bea
\lefteqn{(t^{r-N} e^{-\frac{d_s}{t}})^\prime}\nonumber\\
& = & (r-N)t^{r-N-1}e^{-d_s/t} + t^{r-N} e^{-d_s/t} \dfrac{d_s}{t^2}\\
& = & (r-N)t^{r-N-1}e^{-d_s/t} + t^{r-N-2} e^{-d_s/t} d_s\\
& = & t^{r-N-2}\big((r-N)t + d_s\big) e^{-d_s/t}
\eea
Since $t^{r-N-2}$ and $e^{-d_s/t}$ are always positive, the first derivative $(t^{r-N} e^{-\frac{d_s}{t}})^\prime = 0$ if and only if
\be
t = \dfrac{d_s}{N-r} = \dfrac{\sum_{k=r+1}^N d_k}{N-r}
\ee
and it is positive when $t<\dfrac{\sum_{k=r+1}^N d_k}{N-r}$ and negative otherwise. This means that $\lr(\sigma^2)$ increases for $\sigma^2<\dfrac{\sum_{k=r+1}^N d_k}{N-r}$ and decreases for $\sigma^2>\dfrac{\sum_{k=r+1}^N d_k}{N-r}$. The LR value is maximized when $\sigma^2 = \dfrac{\sum_{k=r+1}^N d_k}{N-r}$. Note that $\dfrac{\sum_{k=r+1}^N d_k}{N-r}$ is the average value of $N-r$ smallest eigenvalues of the sample covariance matrix and in fact a maximum likelihood solution of $\sigma^2$ as shown in the RCML estimator \cite{Kang14}.

\subsection{Proof of Lemma \ref{Lemma3}}
\label{Sec:AppendixC}

For a given rank $r$, the optimal solution of the noise power via the EL approach, $\hat t (=\hat{\sigma}_{\text{EL}}^2)$, is the solution of $\lr(t) = \lr_0$. From Eq. \eqref{Eq:LR_t}, that is, $\hat t$ is the solution of the equation given by
\be
d_p e^{N-r} t^{r-N} e^{-\frac{d_s}{t}} = \lr_0
\ee
Taking $\log$ on both side leads
\be
\label{Eq:OptimalT}
\log d_p + N-r + (r-N)\log t -\frac{d_s}{t} = \log \lr_0
\ee
For simplification, we take substitutions of variables,
\be
\left\{ \begin{array}{l}
a = r - N\\
b = \sum_{k=r+1}^N d_k\\
c = \log \lr_0 - \log\Big (\prod_{k=r+1}^N d_k\Big) + a
\end{array} \right.
\ee
Then, Eq. \eqref{Eq:OptimalT} is simplified to an equation of $t$,
\be
a \log t - \frac{b}{t} = c
\ee
Again, let $u = \log t$. Then, since $t = e^u$, we obtain
\bea
au - be^{-u} = c\\
e^{-u} = \frac{a}{b} u - \frac{c}{b}
\eea
Now let $s=u-\frac{c}{a}$. Then, the equation is
\bea
e^{-s-\frac{c}{a}} = \frac{a}{b} s\\
s e^s = \frac{b}{a} e^{-\frac{c}{a}}\label{Eq:S}
\eea
The solution of Eq. \eqref{Eq:S} is known to be obtained using Lambert $W$ function \cite{Corless96}. That is,
\be
s = W\bigg(\frac{b}{a} e^{-\frac{c}{a}}\bigg)
\ee
where $W(\cdot)$ is a Lambert $W$ function which is defined to be the function satisfying
\be
W(z) e^{W(z)} = z
\ee
Finally, we obtain
\be
u = W\bigg(\frac{b}{a} e^{-\frac{c}{a}}\bigg) + \frac{c}{a}
\ee
and
\be
\hat{\sigma}_{\text{EL}}^2 = \hat t = \exp\Bigg(W\bigg(\frac{b}{a} e^{-\frac{c}{a}}\bigg) + \frac{c}{a}\Bigg)
\ee

\subsection{Proof of Lemma \ref{Lemma4}}
\label{Sec:AppendixD}

We consider 5 cases provided in \cite{Aubry12}.
\begin{enumerate}
    \item $d_1 \leq \sigma^2 \leq \sigma^2 \Kmax$\\
    Since $u^\star = \frac{1}{\Kmax}$,
    \bea
    \lambda_i^\star & = & \min (\min (\Kmax  u^\star,1),\max(u^\star,\frac{1}{\bar d_i}))\\
    & = & \min (\min (1,1),\max(\frac{1}{\Kmax},\frac{1}{\bar d_i}))\\
    & = & \min (1,\frac{1}{\bar d_i}) = 1
    \eea
    Therefore,
    \be
    \hat{\mb R}_{\text{CN}} = \sigma^2 \mb I
    \ee
    and the condition number is $1$.
    \item $\sigma^2 < d_1 \leq \Kmax$\\
    Since $u^\star = \frac{1}{\bar d_1}$,
    \bea
    \lambda_i^\star & = & \min (\min (\Kmax  u^\star,1),\max(u^\star,\frac{1}{\bar d_i}))\\
    & = & \min (\min (\frac{\Kmax}{\bar d_1},1),\max(\frac{1}{\bar d_1},\frac{1}{\bar d_i}))\\
    & = & \min (1,\frac{1}{\bar d_i})\\
    & = & \left\{ \begin{array}{cc} \frac{1}{\bar d_i} & \bar d_i \geq 1\\ 1 & \bar d_i < 1 \end{array} \right.
    \eea
    Therefore,
    \be
    \hat{\mb R}_{\text{CN}} = \hat{\mb R}_{\text{FML}}
    \ee
    and the condition number is $\frac{d_1}{\sigma^2}$.

    \item $d_1 > \sigma^2 \Kmax$ and $u^\star = \frac{1}{\bar d_1}$\\
    Since $u^\star$ is the optimal solution of the optimization problem \eqref{Eq:OptimizationU}, $\frac{dG(u)}{du}|_{u=\frac{1}{\bar d_1}}$ must be zero if $u^\star = \frac{1}{\bar d_1}$. From, Eq. \eqref{Eq:Giu1} and Eq. \eqref{Eq:Giu2}, the first derivative of $G_i(u)$ is given by
    \be
    \label{Eq:Gpiu1}
    G_i^\prime(u) = \left\{ \begin{array}{ll}
    -\frac{1}{u} + \Kmax \bar d_i & \text{if} \quad 0 < u \leq \frac{1}{\Kmax}\\
    0 & \text{if} \quad \frac{1}{\Kmax} \leq u \leq 1 \end{array} \right.
    \ee
    for $\bar d_i \leq 1$, and
    \be
    \label{Eq:Gpiu2}
    G_i^\prime(u) = \left\{ \begin{array}{ll}
    -\frac{1}{u} + \Kmax \bar d_i & \text{if} \quad 0 < u \leq \frac{1}{\Kmax \bar d_i}\\
    0 & \text{if} \quad \frac{1}{\Kmax \bar d_i} < u \leq \frac{1}{\bar d_i}\\
    -\frac{1}{u} + \bar d_i & \text{if} \quad \frac{1}{\bar d_i} \leq u \leq 1 \end{array} \right.
    \ee
    for $\bar d_i > 1$. Therefore,
    \be
   \frac{dG(u)}{du}|_{u=\frac{1}{\bar d_1}} = \sum_{i=\bar N +1}^N (\Kmax \bar d_i - \bar d_1) + \sum_{i=p}^{\bar N} (\Kmax \bar d_i - \bar d_1)
    \ee
    where $p$ is the greatest index such that $\frac{1}{\bar d_1} < \frac{1}{\Kmax \bar d_p}$. For $i=\bar N, \ldots, N$, since $\bar d_i \leq 1$,
    \be
    \Kmax \bar d_i - \bar d_1 < \Kmax - \bar d_1 < 0
    \ee
    and for $i=p, \ldots, \bar N -1$, since $\bar d_1 > \Kmax \bar d_i$, $\Kmax \bar d_i - \bar d_1 < 0$. Therefore, in this case, it is obvious that
    \be
    \frac{dG(u)}{du}|_{u=\frac{1}{\bar d_1}} < 0
    \ee
    which implies $u = \frac{1}{\bar d_1}$ can not be the optimal solution of \eqref{Eq:OptimizationU}.

    \item $d_1 > \sigma^2 \Kmax$ and $u^\star = \frac{1}{\Kmax}$\\
    Aubry \emph{et al.} \cite{Aubry12} showed that $u^\star = \frac{1}{\Kmax}$ if $\frac{dG(u)}{du}|_{u=\frac{1}{\Kmax}} \leq 0$. From Eq. \eqref{Eq:Gpiu1} and Eq. \eqref{Eq:Gpiu2},
    \be
    \frac{dG(u)}{du}|_{u=\frac{1}{\Kmax}} = \sum_{i=\bar N +1}^N \Kmax (\bar d_i - 1) + \sum_{i=1}^p (\bar d_i - \Kmax)
    \ee
    where $p$ is the greatest index such that $\bar d_p > \Kmax$. Therefore,
   \bea
    & \frac{dG(u)}{du}|_{u=\frac{1}{\Kmax}} \leq 0\\
    \Leftrightarrow & \ds\sum_{i=\bar N +1}^N \Kmax (\bar d_i - 1) + \ds\sum_{i=1}^p (\bar d_i - \Kmax) \leq 0\\
    \Leftrightarrow & \Kmax(\sum_{i=\bar N +1}^N (\bar d_i - 1) - p) + \sum_{i=1}^p \bar d_i \leq 0\\
    \Leftrightarrow & \Kmax(\sum_{i=\bar N +1}^N (\bar d_i - 1) - p) \leq -\sum_{i=1}^p \bar d_i\\
    \Leftrightarrow & \Kmax \geq \frac{\sum_{i=1}^p \bar d_i}{p-\sum_{i=\bar N +1}^N (\bar d_i - 1)}
    \eea
    In this case,
    \bea
    \lambda_i^\star & = & \min (\min (\Kmax  u^\star,1),\max(u^\star,\frac{1}{\bar d_i}))\\
    & = & \min (\min (1,1),\max(\frac{1}{\Kmax},\frac{1}{\bar d_i}))\\
    & = & \min (1,\max(\frac{1}{\Kmax},\frac{1}{\bar d_i}))\\
     & = & \left\{ \begin{array}{cc} \min(1,\frac{1}{\Kmax}) & \bar d_i \geq \Kmax\\ \min(1,\frac{1}{\bar d_i}) & \bar d_i < \Kmax \end{array} \right.\\
     & = & \left\{ \begin{array}{cc} \frac{1}{\Kmax} & \bar d_i \geq \Kmax\\ \frac{1}{\bar d_i} & \bar 1 \leq \bar d_i < \Kmax\\ 1 & \bar d_i < 1 \end{array} \right.
    \eea
    Finally we obtain
    \be
    \bs\lambda^\star = \big[ \sigma^2 K_{\max}, \ldots, \sigma^2 K_{\max}, d_{p+1}, \ldots, d_{\bar N}, \sigma^2,\ldots,\sigma^2  \big],
    \ee
     where $p$ and $\bar N$ are the largest indices so that $d_p > \sigma^2 K_{\max}$ and $d_{\bar N} \geq \sigma^2$, respectively.

     \item $d_1 > \sigma^2 K_{\max}$ and $ \Kmax < \frac{\sum_{i=1}^p \bar d_i}{p-\sum_{i=\bar N +1}^N (\bar d_i - 1)}$\\
     In this case, since $\frac{1}{\bar d_1} < u^\star < \frac{1}{\Kmax}$,
      \bea
    \lambda_i^\star & = & \min (\min (\Kmax  u^\star,1),\max(u^\star,\frac{1}{\bar d_i}))\\
    & = & \min (\Kmax  u^\star,\max(u^\star,\frac{1}{\bar d_i}))\\
     & = & \left\{ \begin{array}{cc} \min(\Kmax  u^\star,u^\star) & \bar d_i \geq \frac{1}{u^\star}\\ \min(\Kmax  u^\star,\frac{1}{\bar d_i}) & \bar d_i < \frac{1}{u^\star} \end{array} \right.\\
     & = & \left\{ \begin{array}{cc} u^\star & \bar d_i \geq \frac{1}{u^\star}\\ \frac{1}{\bar d_i} & \frac{1}{\Kmax u^\star} \leq \bar d_i \leq \frac{1}{u^\star}\\ \Kmax u^\star & \bar d_i < \frac{1}{\Kmax u^\star} \end{array} \right.
    \eea
    Therefore, we obtain
     \be
    \bs\lambda^\star = \big[ \frac{\sigma^2}{u^\star}, \ldots, \frac{\sigma^2}{u^\star}, d_{p+1}, \ldots, d_q,\frac{\sigma^2}{u^\star K_{\max}},\ldots,\frac{\sigma^2}{u^\star K_{\max}}  \big]
    \ee
    where $p$ and $q$ are the largest indices so that $d_p > \frac{\sigma^2}{u}$ and $d_q > \frac{\sigma^2}{u \Kmax}$, respectively.
\end{enumerate}

\subsection{Proof of Lemma \ref{Lemma5}}
\label{Sec:AppendixE}

\begin{enumerate}
   \item $d_1\leq \sigma^2$
    \be
    \hat{\mb R}_{\text{CN}} = \sigma^2 \mb I
    \ee
    In this case, $\hat{\mb R}_{\text{CN}}$ does not change, so $\lr(\Kmax)$ is a constant.

    \item $\sigma^2 \leq d_1 \leq \sigma^2 K_{\max}$
    \be
    \hat{\mb R}_{\text{CN}} = \hat{\mb R}_{\text{FML}}
    \ee
    In this case, $\hat{\mb R}_{\text{CN}}$ does not change, so $\lr(\Kmax)$ is a constant.

    \item $d_1 > \sigma^2 K_{\max}$ and $K_{\max} \geq \frac{\sum_{i=1}^p d_i}{c - \sum_{\bar N + 1}^N (d_i -1)}$
    \be
    \hat{\mb R}_{\text{CN}} = \mb\Phi \diag(\bs\lambda^*) \mb\Phi^H
    \ee
    where
    \be
    \bs\lambda^\star = \big[ \sigma^2 K_{\max}, \ldots, \sigma^2 K_{\max}, d_{p+1}, \ldots, d_{\bar N}, \sigma^2,\ldots,\sigma^2  \big],
    \ee
     $p$ and $\bar N$ are the largest indices so that $d_p > \sigma^2 K_{\max}$ and $d_{\bar N} \geq \sigma^2$, respectively.
     \begin{IEEEeqnarray}{rCl}
     \lefteqn{\lr(\Kmax)}\nonumber\\ & = & \frac{\prod_{i=1}^N \frac{d_i}{\lambda_i} e^N }{ \exp (\sum_{i=1}^N \frac{d_i}{\lambda_i}) }\\
     & = & \frac{\ds\prod_{i=1}^p \frac{d_i}{\sigma^2 \Kmax} \cdot \ds\prod_{i=p+1}^{\bar N} 1\cdot \ds\prod_{i=\bar N + 1}^N \frac{d_i}{\sigma^2} \cdot e^N}{ \exp(\ds\sum_{i=1}^p \frac{d_i}{\sigma^2 \Kmax} + \ds\sum_{i=p+1}^{\bar N} 1 + \ds\sum_{i=\bar N + 1}^N \frac{d_i}{\sigma^2} )}\\
     & = & \frac{\prod_{i=1}^p \frac{d_i}{\sigma^2 \Kmax} \cdot \prod_{i=\bar N + 1}^N \frac{d_i}{\sigma^2} \cdot e^N}{ \exp(\ds\sum_{i=1}^p \frac{d_i}{\sigma^2 \Kmax}) \cdot e^{\bar N - p}  \cdot \exp( \ds\sum_{i=\bar N + 1}^N \frac{d_i}{\sigma^2} )}
     \end{IEEEeqnarray}
%     \bea
%     \lefteqn{\lr(\Kmax)}\nonumber\\ & = & \frac{\prod_{i=1}^N \frac{d_i}{\lambda_i} e^N }{ \exp (\sum_{i=1}^N \frac{d_i}{\lambda_i}) }\\
%     & = & \frac{\ds\prod_{i=1}^p \frac{d_i}{\sigma^2 \Kmax} \cdot \ds\prod_{i=p+1}^{\bar N} 1\cdot \ds\prod_{i=\bar N + 1}^N \frac{d_i}{\sigma^2} \cdot e^N}{ \exp(\ds\sum_{i=1}^p \frac{d_i}{\sigma^2 \Kmax} + \ds\sum_{i=p+1}^{\bar N} 1 + \ds\sum_{i=\bar N + 1}^N \frac{d_i}{\sigma^2} )}\\
%     & = & \frac{\prod_{i=1}^p \frac{d_i}{\sigma^2 \Kmax} \cdot \prod_{i=\bar N + 1}^N \frac{d_i}{\sigma^2} \cdot e^N}{ \exp(\ds\sum_{i=1}^p \frac{d_i}{\sigma^2 \Kmax}) \cdot e^{\bar N - p}  \cdot \exp( \ds\sum_{i=\bar N + 1}^N \frac{d_i}{\sigma^2} )}
%     \eea

     \begin{enumerate}
        \item within the range where $p$ remains same
        \begin{IEEEeqnarray}{rCl}
        \lefteqn{\lr(\Kmax)}\nonumber\\  & = & \frac{\prod_{i=1}^p \frac{d_i}{\sigma^2 \Kmax} \cdot \prod_{i=\bar N + 1}^N \frac{d_i}{\sigma^2} \cdot e^N}{ \exp(\ds\sum_{i=1}^p \frac{d_i}{\sigma^2 \Kmax}) \cdot e^{\bar N - p}  \cdot \exp( \ds\sum_{i=\bar N + 1}^N \frac{d_i}{\sigma^2} )}\IEEEeqnarraynumspace\\
         & = & c_1 \frac{\prod_{i=1}^p \frac{d_i}{\sigma^2 \Kmax}}{\exp(\sum_{i=1}^p \frac{d_i}{\sigma^2 \Kmax})}\\
        & = & c_1 \frac{\frac{1}{(\sigma^2 \Kmax)^p} \prod_{i=1}^p d_i}{\exp( \frac{1}{\sigma^2 \Kmax} \sum_{i=1}^p d_i)}\\
        & = & c_1 \frac{\frac{1}{(\sigma^2 \Kmax)^p} \prod_{i=1}^p d_i}{(\exp( \sum_{i=1}^p d_i))^\frac{1}{\sigma^2 \Kmax}}\\
        & = & c_2 \frac{ (\frac{1}{\Kmax})^p}{c_3^\frac{1}{\Kmax}}\\
        & = & c_2 \frac{1}{(\Kmax)^p \cdot c_3^\frac{1}{\Kmax}}\label{Eq:LRKmax}
        \end{IEEEeqnarray}
  %      \bea
%        \lefteqn{\lr(\Kmax)}\nonumber\\  & = & \frac{\prod_{i=1}^p \frac{d_i}{\sigma^2 \Kmax} \cdot \prod_{i=\bar N + 1}^N \frac{d_i}{\sigma^2} \cdot e^N}{ \exp(\ds\sum_{i=1}^p \frac{d_i}{\sigma^2 \Kmax}) \cdot e^{\bar N - p}  \cdot \exp( \ds\sum_{i=\bar N + 1}^N \frac{d_i}{\sigma^2} )}\\
%        & = & c_1 \frac{\prod_{i=1}^p \frac{d_i}{\sigma^2 \Kmax}}{\exp(\sum_{i=1}^p \frac{d_i}{\sigma^2 \Kmax})}\\
%        & = & c_1 \frac{\frac{1}{(\sigma^2 \Kmax)^p} \prod_{i=1}^p d_i}{\exp( \frac{1}{\sigma^2 \Kmax} \sum_{i=1}^p d_i)}\\
%        & = & c_1 \frac{\frac{1}{(\sigma^2 \Kmax)^p} \prod_{i=1}^p d_i}{(\exp( \sum_{i=1}^p d_i))^\frac{1}{\sigma^2 \Kmax}}\\
%        & = & c_2 \frac{ (\frac{1}{\Kmax})^p}{c_3^\frac{1}{\Kmax}}\\
%        & = & c_2 \frac{1}{(\Kmax)^p \cdot c_3^\frac{1}{\Kmax}}\label{Eq:LRKmax}
%       \eea
        where $c_1 = \frac{\prod_{i=\bar N + 1}^N \frac{d_i}{\sigma^2} \cdot e^N}{\exp (\bar N - p)  \cdot \exp( \sum_{i=\bar N + 1}^N \frac{d_i}{\sigma^2} )}$, $c_2 = c_1 \frac{ \prod_{i=1}^p d_i}{\sigma^{2p}}$, and $c_3 = \exp( \frac{1}{\sigma^2}\sum_{i=1}^p d_i)$.\\
        Now let's evaluate the first derivative of the denominator of Eq. \eqref{Eq:LRKmax}.
        \begin{IEEEeqnarray}{rCl}%\bea
        \lefteqn{((\Kmax)^p \cdot c_3^\frac{1}{\Kmax})^\prime}\nonumber\\ & = & p (\Kmax)^{p-1} c_3^\frac{1}{\Kmax} + (\Kmax)^p \frac{c_3^\frac{1}{\Kmax} \log c_3}{-(\Kmax)^2}\\
        & = & p (\Kmax)^{p-1} c_3^\frac{1}{\Kmax} - (\Kmax)^{p-2} c_3^\frac{1}{\Kmax} \log c_3\IEEEeqnarraynumspace\\
        & = & (\Kmax)^{p-2} c_3^\frac{1}{\Kmax} (p\Kmax- \log c_3)\\
        & = & (\Kmax)^{p-2} c_3^\frac{1}{\Kmax} (p\Kmax- \frac{1}{\sigma^2}\sum_{i=1}^p d_i)
        \end{IEEEeqnarray}%\eea
        Since $d_1 > d_2 > \cdots > d_p > \sigma^2 \Kmax$,
        \be
        p\Kmax- \frac{1}{\sigma^2}\sum_{i=1}^p d_i < 0
        \ee
        This implies the denominator of Eq. \eqref{Eq:LRKmax} is a decreasing function, and therefore, $LR(\Kmax)$ is a increasing function with respect to $\Kmax$.

        \item $p \rightarrow p+1$ as $\Kmax$ decreases\\
        The $\lr(\Kmax)$ is a continuous function since $\lambda_{p+1} = d_{p+1}$ at the moment that $\sigma^2 \Kmax = d_{p+1}$ and there is no discontinuity of $\lambda_i$. Therefore, $\lr(\Kmax)$ is an increasing function in this case.
    \end{enumerate}

    \item $d_1 > \sigma^2 K_{\max}$ and $K_{\max} < \frac{\sum_{i=1}^c d_i}{c - \sum_{\bar N + 1}^N (d_i -1)}$\\
    \be
    \hat{\mb R}_{\text{CN}} = \mb\Phi \diag(\bs\lambda^*) \mb\Phi^H
    \ee
    where
    \be
    \bs\lambda^\star = \big[ \frac{\sigma^2}{u}, \ldots, \frac{\sigma^2}{u}, d_{p+1}, \ldots, d_q,\frac{\sigma^2}{u K_{\max}},\ldots,\frac{\sigma^2}{uK_{\max}}  \big]
    \ee
    $p$, $q$, and $\bar N$ are the vector of the eigenvalues of the estimate, the largest indices so that $d_p > \frac{\sigma^2}{u}$, $d_q > \frac{\sigma^2}{u \Kmax}$, and $d_{\bar N} \geq \sigma^2$, respectively.

    Before we prove the increasing property of $\lr(\Kmax)$, we show $u$ decreases as $\Kmax$ increases. $u$ is the optimal solution of the optimization problem. In this case, $u^\star$, the optimal solution of the optimization problem \eqref{Eq:OptimizationU} is obtained by making the first derivative of the cost function 0. Let $u_1$ and $u_2$ be the optimal solutions for $\Kmax_1$ and $\Kmax_2$, respectively. Then, $\sum_{i=1}^N G_i^\prime(u_1) = 0$ for $\Kmax_1$. Since $\frac{1}{d_i} \leq u_1 \leq \frac{1}{\Kmax_1}$ in this case, for $\Kmax_2 < \Kmax_1$, the value of $G_i^\prime (u_1)$ decreases for $d_i \leq 1$. $G_i^\prime(u)$ also decreases for $d_i >1$ and $u \leq \frac{1}{\Kmax d_i}$ and remain same for $d_i > 1$ and $\frac{1}{\Kmax d_i} < u$. Therefore, $\sum_{i=1}^N G_i^\prime(u_1) < 0$ for $\Kmax_2$. Finally, since $\sum_{i=1}^N G_i^\prime(u_2)$ must be zero for $\Kmax_2$, it is obvious that $u_1 < u_2$. This shows that $u$ decreases as $\Kmax$ increases.

    Now we show the increasing property of $\lr(\Kmax)$.
    \begin{enumerate}
        \item within the range where $p$ and $q$ remain same\\
        In this case, We show $\lr(u)$ is a decreasing function of $u$ and an increasing function of $\Kmax$ for each of $u$ and $\Kmax$.
        \begin{enumerate}
        \item Proof of $\lr(u)$ is a decreasing function.
        \biea
        \lefteqn{\lr(u)}\nonumber\\ & = & \frac{\prod_{i=1}^p \frac{u d_i}{\sigma^2} \cdot \prod_{i=q+1}^{\bar N} \frac{\Kmax u d_i}{\sigma^2} \cdot e^N}{ \exp(\sum_{i=1}^p \frac{u d_i}{\sigma^2} + \sum_{i=p+1}^{q} 1}\nonumber\\
         && \: \frac{}{+ \sum_{i=q + 1}^N \frac{\Kmax u d_i}{\sigma^2} )}\\
        & = & \frac{u^p \prod_{i=1}^p \frac{d_i}{\sigma^2} \cdot u^{N-q} \prod_{i=q+1}^{\bar N} \frac{\Kmax d_i}{\sigma^2} \cdot e^N}{ \exp( u (\sum_{i=1}^p \frac{d_i}{\sigma^2} + \sum_{i=q+1}^N \frac{\Kmax d_i}{\sigma^2})}\nonumber\\
        && \: \frac{}{ + q - p )}\\
        & = & \frac{c_1 u^{N-q+p}}{ \exp( c_2 u + c_3 )}\\
        & = & c_4 \frac{u^{N-q+p}}{c_5^u}\label{Eq:LR_u}
        \eiea
        where $c_1 = \prod_{i=1}^p \frac{d_i}{\sigma^2} \cdot \prod_{i=q+1}^{\bar N} \frac{\Kmax d_i}{\sigma^2}\cdot e^N$, $c_2 = \sum_{i=1}^p \frac{d_i}{\sigma^2} + \sum_{i=q+1}^N \frac{\Kmax d_i}{\sigma^2}$, $c_3 = q-p$, $c_4 = \frac{c_1}{e^{c_3}}$, and $c_5 = e^{c_2}$.
        The first derivative of Eq. \eqref{Eq:LR_u} is obtained by
        \biea
        \lefteqn{\lr^\prime(u)}\nonumber\\ & = & (N-q+p) u^{N-q+p-1}c_5^{-u}\nonumber\\
        &&  - \:  u^{N-q+p} \log c_5 \cdot c_5^{-u}\\
        & = & u^{N-q+p-1}c_5^{-u}(N-q+p  - u\log c_5)\IEEEeqnarraynumspace\\
        & = & u^{N-q+p-1}c_5^{-u}(N-q+p  - c_2 u)\\
        & = & u^{N-q+p-1}c_5^{-u}(N-q+p\nonumber\\
        && - \:  u(\sum_{i=1}^p \frac{d_i}{\sigma^2} + \sum_{i=q+1}^N \frac{\Kmax d_i}{\sigma^2}))
        \eiea
        Since $\frac{\sigma^2}{u} \leq d_p$,
        \biea
        \lefteqn{N-q+p  - u(\sum_{i=1}^p \frac{d_i}{\sigma^2} + \sum_{i=q+1}^N \frac{\Kmax d_i}{\sigma^2})}\nonumber\\ & \leq & N-q+p - u(\frac{p}{u} + \frac{N-q}{u}\cdot \Kmax)\IEEEeqnarraynumspace\\
        & = & N - q - \Kmax(N-q)
        \eiea
        Since $\Kmax > 1$, $\lr^\prime(u)<0$ which implies $\lr(u)$ is a decreasing function with respect to $u$.

        \item Proof of $\lr(\Kmax)$ is an increasing function.
        \biea
        \lefteqn{\lr(\Kmax)}\nonumber\\ & = & \frac{\prod_{i=1}^p \frac{u d_i}{\sigma^2} \cdot \prod_{i=q+1}^{\bar N} \frac{\Kmax u d_i}{\sigma^2} \cdot e^N}{ \exp(\sum_{i=1}^p \frac{u d_i}{\sigma^2} + \sum_{i=p+1}^{q} 1}\nonumber\\
         && \frac{}{+\sum_{i=q + 1}^N \frac{\Kmax u d_i}{\sigma^2} )}\IEEEeqnarraynumspace\\
        & = & \frac{c_1 \Kmax^{N-q}}{ \exp( c_2 \Kmax + c_3)}\\
        & = & c_4 \frac{\Kmax^{N-q}}{c_5^{\Kmax}}\label{Eq:LR_Kmax}
        \eiea
        where $c_1 = \prod_{i=1}^p \frac{u d_i}{\sigma^2} \cdot \prod_{i=q+1}^{\bar N} \frac{u d_i}{\sigma^2} \cdot e^N$, $c_2 = \sum_{i=q+1}^N \frac{ud_i}{\sigma^2}$, $c_3 = \sum_{i=1}^p \frac{u d_i}{\sigma^2} + q - p$, $c_4 = \frac{c_1}{e^{c_3}}$, and $c_5 = e^{c_2}$. The first derivative is
      \bea
        \lefteqn{\lr^\prime(\Kmax)}\nonumber\\ & = & (N-q) \Kmax^{N-q-1}c_5^{-\Kmax} \nonumber\\
        && -\: \Kmax^{N-q} \log c_5 \cdot c_5^{-\Kmax}\\
        & = & \Kmax^{N-q-1}\nonumber\\
        && \times \: c_5^{-\Kmax}(N-q  - \Kmax\log c_5)\\
        & = & \Kmax^{N-q+p-1}\nonumber\\
        && \times \: c_5^{-u}(N-q  - c_2 \Kmax)\\
        & = & \Kmax^{N-q+p-1}\nonumber\\
        && \times \: c_5^{-u}(N-q  - \Kmax \sum_{i=q+1}^N \frac{ud_i}{\sigma^2})\IEEEeqnarraynumspace
        \eea
        Since $\frac{\sigma^2}{u\Kmax} \leq d_{q+1}$,
        \bea
        \lefteqn{N-q  -\Kmax \sum_{i=q+1}^N \frac{ud_i}{\sigma^2}}\nonumber\\ & \geq & N-q - \Kmax (\frac{N-q}{\Kmax}) = 0
        \eea
        Therefore, $\lr^\prime(\Kmax) \geq 0$ and $\lr(\Kmax)$ is an increasing function with respect to $\Kmax$.
        \end{enumerate}
        These two proofs show that $\lr(u,\Kmax)$ is an increasing function with respect to $\Kmax$.

   \item $p$ and $q$ changes as $\Kmax$ decreases\\
        The $\lr(u, \Kmax)$ is a continuous function, and therefore, $\lr(u,\Kmax)$ is an increasing function in this case.
    \end{enumerate}

\end{enumerate} 

%\section{Conclusion}
%The conclusion goes here.

%%%% References
\scriptsize
\bibliographystyle{C:/Users/BXK265/BOXSYN\string~2/BOXSYN\string~1/LaTex/Bibliography/IEEEbib}
\bibliography{C:/Users/BXK265/BOXSYN\string~2/BOXSYN\string~1/LaTex/Bibliography/IEEEabrv,C:/Users/BXK265/BOXSYN\string~2/BOXSYN\string~1/LaTex/Bibliography/Bosung}
\end{document}